\documentclass[letter,11pt]{article}
\usepackage[utf8]{inputenc}
\usepackage[margin=1in]{geometry}
\usepackage{complexity}

\usepackage[T1]{fontenc}
\usepackage{soul} 
\usepackage{proba} 
\usepackage{amssymb}
\usepackage{amsmath}
\usepackage{amsthm}
\usepackage{amsfonts}
\usepackage{mathtools}
\usepackage{xspace}
\usepackage{bm}
\usepackage{nicefrac}
\usepackage{commath}
\usepackage{bbm}
\usepackage{boxedminipage}
\usepackage{xparse}
\usepackage{xcolor}
\usepackage{float}
\usepackage{multirow}
\usepackage{graphicx}
\usepackage{caption}
\usepackage{subcaption}
\usepackage{ifthen}
\usepackage[noend]{algpseudocode}
\usepackage[boxed]{algorithm}
\usepackage{colortbl} 
\usepackage{fancyhdr}   
\usepackage{hyperref}
    \hypersetup{ colorlinks=true, linkcolor=blue, filecolor=magenta, urlcolor=blue,}
\usepackage{fullpage}
\usepackage[utf8]{inputenc}
\usepackage{environ}
\usepackage{mdframed}
\usepackage{float,graphicx,verbatim,hyperref,amssymb,amsmath,amsthm,enumerate,multicol, xspace,xcolor,mathtools,thmtools,thm-restate,cleveref,xspace,tikz,caption,subcaption,wasysym, enumitem}
\newtheorem{lemma}{Lemma}

\newtheorem{theorem}{Theorem}
\newtheorem{corollary}{Corollary}
\newtheorem{definition}{Definition}

\newtheorem{observation}{Observation}

\newtheorem{claim}{Claim}

\newcommand{\tuple}[1]{\ensuremath{\left(#1\right)}}

\newcommand{\ceil}[1]{\ensuremath{\left\lceil{#1}\right\rceil}\xspace}
\newcommand{\floor}[1]{\ensuremath{\left\lfloor{#1}\right\rfloor}\xspace}

    \newcommand{\cT}{\ensuremath{{\mathcal T}}\xspace}

\newcommand{\permutationGraph}{\ensuremath{\mathsf{PermDAG}}\xspace}

\newcommand{\GreedyAssign}{\ensuremath{\mathtt{GreedyAssign}}\xspace}

\newcommand{\sequence}{\ensuremath{\sigma}\xspace}
\newcommand{\optsequence}{\ensuremath{\sigma^*}\xspace}

\newcommand{\tp}{\tuple}
 
\def\*#1{\mathbf{#1}}
\def\+#1{\mathcal{#1}}

\DeclareMathOperator*{\maxx}{max}

\renewcommand{\vec}[1]{\boldsymbol{\mathrm{#1}}}

\NewEnviron{problem}[1]{%
	\begin{center}\fbox{\parbox{5in}{%
				{\centering\scshape #1\par}%
				\parskip=1ex
				\everypar{\hangindent=1em}%
				\BODY
}}\end{center}}

\newcommand{\umaxBT}{\textsc{Undir\-Max\-Binary\-Tree}\xspace}
\newcommand{\rumaxBT}{\textsc{rooted-\-Undir\-Max\-Binary\-Tree}\xspace}

\newcommand{\mbt}{\textsf{MBT}}

\tikzstyle{vertex}=[circle, draw, inner sep=1pt, minimum size=18pt]

\usetikzlibrary{arrows.meta}
\tikzset{>={Latex[width=2mm,length=2mm]}}

\def\final{0}  % set this to 1 to get a comment-free version
\def\iflong{\iffalse}
\ifnum\final=0  %namely if we allow comments in the output
\newcommand{\knote}[1]{{\color{red}[{\tiny Karthik: \bf #1}]\marginpar{\color{red}*}}}
\newcommand{\snote}[1]{{\color{blue}[{\tiny Shubhang: \bf #1}]\marginpar{\color{blue}*}}}
\newcommand{\mnote}[1]{{\color{purple}[{Minshen: \bf #1}]\marginpar{\color{red}*}}}

\newcommand{\todonote}[1]{{\color{red}[{\tiny TODO: \bf #1}]\marginpar{\color{red}*}}}
\newcommand{\enote}[1]{{\color{green}[{ Elena: \bf #1}]\marginpar{\color{green}*}}}

\else % in this case [final=1] we don't want any comments to show
\newcommand{\knote}[1]{}
\newcommand{\snote}[1]{}
\newcommand{\mnote}[1]{}
\newcommand{\enote}[1]{}
\newcommand{\todonote}[1]{}
\fi  %ok, we're done with defining comment macros
\title{Fixed-Parameter Algorithms for Longest Heapable Subsequence and Maximum Binary Tree\thanks{Karthik was supported in part by NSF CCF-1814613 and NSF CCF-1907937. Elena, Young-San, and Minshen were supported in part by NSF CCF-1910659 and NSF CCF-1910411. }
}

\author{Karthekeyan Chandrasekaran\thanks{University of Illinois, Urbana-Champaign, Email: \{karthe, smkulka2\}@illinois.edu}
\and Elena Grigorescu\thanks{Purdue University, Email: \{elena-g, lin532, zhu628\}@purdue.edu}
\and Gabriel Istrate\thanks{West University of Timi\c{s}oara, Romania, and the e-Austria Research Institute. Email: gabrielistrate@acm.org}\\
\and Shubhang Kulkarni\footnotemark[1]
\and Young-San Lin\footnotemark[2]
\and Minshen Zhu\footnotemark[2]
}
\date{\today}

\begin{document}
\maketitle
%\tableofcontents
%\newpage
\begin{abstract}
  A heapable sequence is a sequence of numbers that can be arranged in  a \emph{min-heap data structure}.
  Finding a longest heapable subsequence of a given sequence was proposed by Byers, Heeringa, Mitzenmacher, and Zervas (ANALCO 2011) as a generalization of the well-studied longest increasing subsequence problem and its complexity still remains open. An equivalent formulation of the longest heapable subsequence problem is that of finding a maximum-sized binary tree in a given permutation directed acyclic graph (permutation DAG). In this work, we study parameterized algorithms for both longest heapable subsequence as well as maximum-sized binary tree. We show the following results:
  \begin{enumerate}
      \item The longest heapable subsequence problem can be solved in $k^{O(\log{k})}n$ time, where $k$ is the number of distinct values in the input sequence. 
      \item We introduce the \emph{alphabet size} as a new parameter in the study of computational problems in permutation DAGs. Our result on longest heapable subsequence implies that the maximum-sized binary tree problem in a given permutation DAG is fixed-parameter tractable when parameterized by the alphabet size. 
      \item We show that the alphabet size with respect to a fixed topological ordering can be computed in polynomial time, admits a min-max relation, and has a polyhedral description.
      \item We design a fixed-parameter algorithm with run-time $w^{O(w)}n$ for the maximum-sized binary tree problem in undirected graphs when parameterized by treewidth $w$. 
  \end{enumerate}
  Our results make progress towards understanding the complexity of the longest heapable subsequence and maximum-sized binary tree in permutation DAGs from the perspective of parameterized algorithms. We believe that the parameter alphabet size that we introduce is likely to be useful in the context of optimization problems defined over permutation DAGs. 
  
\end{abstract}

\section{Introduction}

The longest increasing subsequence is a fundamental computational problem that has led to numerous discoveries in algorithms as well as combinatorics. The motivation behind this work is a generalization of the longest increasing subsequence problem, known as the \emph{longest heapable subsequence} problem, introduced by Byers, Heeringa, Mitzenmacher, and Zervas \cite{byers2010heapable}. We begin by defining this problem. A rooted tree whose nodes are labeled with values has the \emph{heap property} if the value of every node is at least that of its parent; a sequence of natural numbers is \emph{heapable} if the elements can be sequentially placed one at a time to form a binary tree with the heap property. For example, the sequence $1, 5, 3, 2, 4$ is not heapable while the sequence $1, 3, 3, 2, 4$ is heapable. Throughout this work, we will be interested in sequences whose elements are natural numbers. 
In the longest heapable subsequence problem, the goal is to find a longest heapable subsequence of a given sequence. Although the longest \emph{increasing} subsequence problem is solvable in polynomial-time, the complexity  of the longest \emph{heapable} subsequence problem is still open. 

The problem of verifying if a given sequence is heapable, although non-trivial, is solvable efficiently using a greedy approach \cite{byers2010heapable}. 

In order to address 
the longest heapable subsequence problem, Porfilio \cite{porfilio2015combinatorial} observed a connection to a graph problem on directed acyclic graphs (DAGs). The permutation DAG associated with a sequence $\sigma=(\sigma(1), \sigma(2), \ldots, \sigma(n))$, denoted $\permutationGraph(\sigma)$, is obtained by introducing a vertex $t_i$ for every sequence element $i\in [n]$, and arcs 
$(t_j, t_i)$ for every $i, j\in [n]$ such that $i<j$ and $\sigma(i) \le \sigma(j)$. 
We recall that a directed graph $G$ is a \emph{permutation DAG} if there exists a sequence $\tau$ such that $G$ is isomorphic to $\permutationGraph(\tau)$. 
We need the notion of a binary tree in a given directed graph $G$: a subgraph $T$ of $G$ is an $r$-rooted binary tree if $r$ is the unique vertex in $T$ with no outgoing edges, every vertex in $T$ has a unique directed path to $r$ in $T$, and every vertex in $T$ has in-degree at most $2$ in $T$; the size of $T$ is the number of vertices in $T$. 
Porfilio showed that a longest heapable subsequence of a given sequence $\sigma$ is equivalent to a maximum-sized binary tree in $\permutationGraph(\sigma)$. 
This result raises the question of whether one can efficiently find a maximum-sized binary tree in a given permutation DAG. The complexity of this problem also remains open. 

In an earlier work \cite{chandrasekaran2019maximum}, we showed that maximum-sized binary tree in arbitrary input directed graphs  is fixed-parameter tractable when parameterized by the solution size: we gave a $2^{k}n^{O(1)}$ time algorithm, where $k$ is the size of the largest binary tree and $n$ is the number of vertices in the input graph. This also implies that the longest heapable subsequence problem is fixed-parameter tractable when parameterized by the solution size. In this work, we consider two alternative parameterizations for the maximum-sized binary tree/longest heapable subsequence problem. 

Firstly, we show that the longest heapable subsequence problem is fixed-parameter tractable when parameterized by the number of distinct values in the input sequence. Next, we introduce \emph{alphabet size} as a new parameter in the study of computational problems in permutation DAGs. Our algorithmic result for longest heapable subsequence problem implies that the maximum-sized binary tree problem in a given permutation DAG is fixed-parameter tractable when parameterized by the alphabet size. We currently do not know how to compute the alphabet size of a given permutation DAG. As a stepping stone towards computing alphabet size, we show that alphabet size \emph{with respect to a fixed topological ordering}  can be computed efficiently and it also admits a min-max relation and a polyhedral description. Our results suggest that alphabet size is an interesting parameterization for computational problems defined on permutation DAGs and merits a thorough study. Finally, we design a fixed-parameter algorithm for the maximum-sized binary tree problem in undirected graphs when parameterized by treewidth. We elaborate on our contributions now. 

\subsection{Results}

Our first result shows that the longest heapable subsequence problem is fixed-parameter tractable when parameterized by the number of distinct values in the sequence. 
\begin{restatable}{theorem}{mbtAlphabetK} \label{thm:MBT-for-small-alphabet-size}
There exists an algorithm that takes as input an $n$-length sequence $\tau$ with $k$ distinct values 
and returns a longest heapable subsequence of $\tau$ in time $(k+1)!\cdot k\cdot O(n)$. 
Equivalently, our algorithm returns a maximum-sized binary tree in $\permutationGraph (\tau)$. 

\end{restatable}

We emphasize that our algorithm also works in the streaming model of computation---i.e., when the input sequence arrives one by one and the algorithm has to find the longest heapable subsequence of the input that has arrived so far (in particular, the algorithm is not allowed to store the entire input sequence). The space complexity of our algorithm is $(k+1)! \cdot k \cdot O(\log n)$ and is logarithmic for constant $k$. 

Theorem \ref{thm:MBT-for-small-alphabet-size} can also be viewed as a fixed-parameter algorithm to find a maximum-sized binary tree in a given permutation DAG when parameterized by \emph{alphabet size}. We define this parameter now. 
We note that for a fixed permutation DAG $G$, there could be several sequences $\tau$ such that $G(\tau)$ is isomorphic to $G$. This motivates our parameterization for permutation DAGs: The \emph{alphabet size} of a permutation DAG $G$, denoted $\alpha(G)$, is defined as follows (see Figure \ref{fig:alphabet-size} for an example): 
\[
\alpha(G):= \min \{k:\exists \ \tau \in [k]^n \text{ with } \permutationGraph(\tau) \text{ being isomorphic to } G\}.
\]
We recall that directed graphs $G=(V,A)$ and $G'=(V,',A')$ are isomorphic if there exists a bijection $\phi:V'\rightarrow V$ such that $(u',v')\in A'$ if and only if $(\phi(u'), \phi(v'))\in A$. 
Theorem \ref{thm:MBT-for-small-alphabet-size} also implies that there exists an algorithm that takes as input, an $n$-vertex permutation DAG $G$ and a sequence $\tau$ with $\alpha(G)=k$ distinct values such that $G(\tau)$ is isomorphic to $G$ and returns a maximum-sized binary tree in $G$ in time $(k+1)!\cdot k\cdot O(n)$.
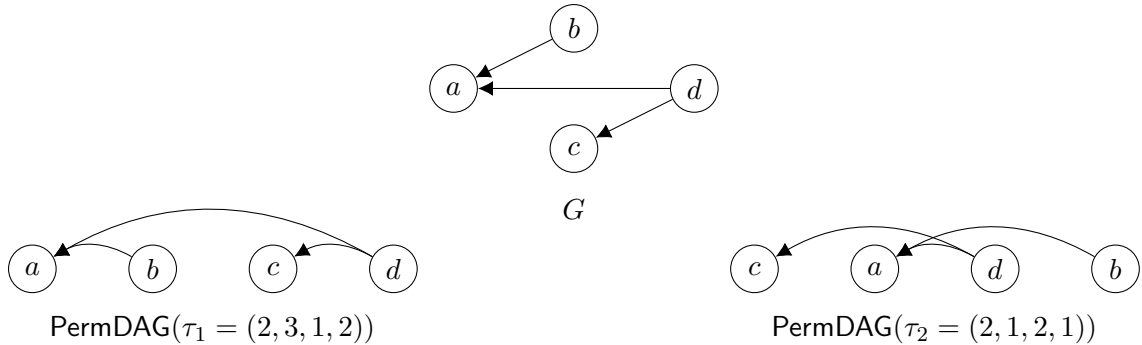
\begin{figure}[ht]
\begin{center}
\begin{tikzpicture}[scale =0.8]
\coordinate (O) at (0,0);

\node [vertex] at (2, 0)  (1) {$a$};
\node [vertex] at (4, 0)  (2) {$b$};
\node [vertex] at (6, 0)  (3) {$c$};
\node [vertex] at (8, 0)  (4) {$d$};
\node [vertex] at (14, 0) (5) {$c$};
\node [vertex] at (16, 0) (6) {$a$};
\node [vertex] at (18, 0) (7) {$d$};
\node [vertex] at (20, 0) (8) {$b$};

\node [vertex] at (9, 3)  (a) {$a$};
\node [vertex] at (11, 4)  (b) {$b$};
\node [vertex] at (11, 2)  (c) {$c$};
\node [vertex] at (13, 3)  (d) {$d$};

\node at (5,-1) {$\permutationGraph(\tau_1=(2,3,1,2))$};
\node at (17,-1) {$\permutationGraph(\tau_2=(2,1,2,1))$};
\node at (11,1) {$G$};

%\draw[->, bend right] (3) to (2);
\draw[->, bend right] (4) to (3);
%\draw[->, bend right] (6) to (3);
%\draw[->, bend right] (7) to (1);
\draw[->, bend right] (2) to (1);
\draw[->, bend right] (4) to (1);

\draw[->, bend right] (7) to (5);
\draw[->, bend right] (7) to (6);
\draw[->, bend right] (8) to (6);

\draw[->] (b) to (a);
\draw[->] (d) to (a);
\draw[->] (d) to (c);

\end{tikzpicture}
\end{center}

\caption{Let $G$ be the input permutation DAG. The graph $G$ is isomorphic to $\permutationGraph(\tau_1)$ and $\permutationGraph(\tau_2)$. The sequence $\tau_1=(2,1,2,1)$ uses only two distinct values, which turns out to be the minimum, so $\alpha(G)=2$. %Let $\gamma_1=(a,b,c,d)$ and $\gamma_2=(c,a,d,b)$ be two ascending topological orderings, then $\alpha(G,\gamma_1)=3$ and $\alpha(G,\gamma_2)=2$.
}
\label{fig:alphabet-size}
\end{figure}

Next, we explore algorithmic aspects of our newly defined parameter, namely the alphabet size. A natural question is whether the alphabet size of a given permutation DAG can be computed in polynomial-time. Currently, we do not know the answer to this question. However, there is a natural related problem that seems like a stepping stone towards resolving the complexity of computing the alphabet size of permutation DAGs. We define this related problem now. 

\medskip
We recall that every DAG $G=(V,A)$ admits a topological ordering---a bijection $\gamma:V\rightarrow [n]$ corresponding to a permutation of its $n$ vertices such that every arc $(v,u)\in A$ has $\gamma(u)<\gamma(v)$ (i.e., all edges are oriented in the backward direction with respect to the ordering defined by $\gamma$). 
For a fixed topological ordering $\gamma:V\rightarrow [n]$ of an $n$-vertex permutation DAG $G$, we  define the \emph{$\gamma$-alphabet size of $G$}, denoted $\alpha(G,\gamma)$, as follows (see Figure \ref{fig:seq-alphabet-size} for an example illustrating the definition):
\begin{align*}
\alpha(G, \gamma)
&:=\min \left\{k:\exists \ \tau \in [k]^n \text{ with } \permutationGraph(\tau)=(\{t_1, \ldots, t_n\}, A') \text{ being isomorphic to $G$}\right.\\
& \quad \quad \quad \quad \quad \left.\text{under the mapping $\phi: \{t_1, \ldots, t_n\}\rightarrow V(G)$ given by $\phi(t_i)=\gamma^{-1}(i)\ \forall \ i\in [n]$}\right\}. 
\end{align*}

We note that the optimization problem $\alpha(G, \gamma)$ may be infeasible in which case, we use $\alpha(G, \gamma)=\infty$ as the convention. 
%We will later see that an umbrella-free topological ordering $\gamma$ of $G$ guarantees that the problem is feasible. 
The following relationship between alphabet size and $\gamma$-alphabet size is immediate for a permutation DAG $G$: 
\[
\alpha(G)=\min \{\alpha(G,\gamma): \gamma \text{ is a topological ordering of } G\}.
\]

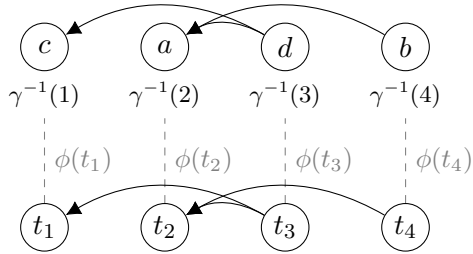
\begin{figure}[ht]
\begin{center}
\begin{tikzpicture}[scale =0.8]
\coordinate (O) at (0,0); 

\node [vertex] at (10, 3) (1) {$c$};
\node [vertex] at (12, 3) (2) {$a$};
\node [vertex] at (14, 3) (3) {$d$};
\node [vertex] at (16, 3) (4) {$b$};
\node  at (10, 2.2) (g1) {\footnotesize$\gamma^{-1}(1)$};
\node  at (12, 2.2) (g2) {\footnotesize$\gamma^{-1}(2)$};
\node  at (14, 2.2) (g3) {\footnotesize$\gamma^{-1}(3)$};
\node  at (16, 2.2) (g4) {\footnotesize$\gamma^{-1}(4)$};
%\node [vertex] at (10, 3) (5) {$2$};
%\node [vertex] at (12, 3) (6) {$1$};
%\node [vertex] at (14, 3) (7) {$2$};
%\node [vertex] at (16, 3) (8) {$1$};

\node [vertex] at (10, 0) (t1) {$t_1$};
\node [vertex] at (12, 0) (t2) {$t_2$};
\node [vertex] at (14, 0) (t3) {$t_3$};
\node [vertex] at (16, 0) (t4) {$t_4$};

%\node at (13,5) {the input permutation DAG $G$ according to the topological ordering $\gamma$};

%\node at (13,2) {arc orientation according to $\tau=(2,1,2,1)$};

%\node at (13,-1) {$\permutationGraph(\tau=(2,1,2,1))$};

\iffalse
\draw[dashed, gray, label=$x$] (t1) to (g1);
\draw[dashed, gray] (t2) to (g2);
\draw[dashed, gray] (t3) to (g3);
\draw[dashed, gray] (t4) to (g4);
\fi

\path[-]
    (t1) edge [dashed, gray] node [right] {\small $\phi(t_1)$} (g1)
    (t2) edge [dashed, gray] node [right] {\small$\phi(t_2)$} (g2)
    (t3) edge [dashed, gray] node [right] {\small$\phi(t_3)$} (g3)
    (t4) edge [dashed, gray] node [right] {\small$\phi(t_4)$} (g4);

\draw[->, bend right] (3) to (1);
\draw[->, bend right] (3) to (2);
\draw[->, bend right] (4) to (2);

%\draw[->, bend right] (7) to (5);
%\draw[->, bend right] (7) to (6);
%\draw[->, bend right] (8) to (6);

\draw[->, bend right] (t3) to (t1);
\draw[->, bend right] (t3) to (t2);
\draw[->, bend right] (t4) to (t2);

\end{tikzpicture}
\end{center}

\caption{
The graph at the top corresponds to $G$ in topological order $\gamma$ where $\gamma(c)=1$, $\gamma(a)=2$, $\gamma(d)=3$, and $\gamma(b)=4$. The graph at the bottom corresponds to $\permutationGraph(\tau=(2,1,2,1))$. Note that the two graphs are isomorphic under the mapping $\phi:\{t_1, t_2, t_3, t_4\}\rightarrow V(G)$ given by $\phi(t_1)=\gamma^{-1}(1)=c$, $\phi(t_2)=\gamma^{-1}(2)=a$, $\phi(t_3)=\gamma^{-1}(3)=d$, and $\phi(t_4)=\gamma^{-1}(4)=b$ (shown by dotted lines). We have that $\alpha(G,\gamma)=2$ and is achieved by the sequence $\tau$.
%Let $G$ be the input permutation DAG. Let $\gamma: V(G) \to [4]$ be such that $\gamma(a)=2$, $\gamma(b)=4$, $\gamma(c)=1$, and $\gamma(d)=3$. Let $\tau=(2,1,2,1)$. The isomorphic mapping $\phi$ is such that $\phi(t_1)=\gamma^{-1}(1)=c$, $\phi(t_2)=\gamma^{-1}(2)=a$, $\phi(t_3)=\gamma^{-1}(3)=d$, and $\phi(t_4)=\gamma^{-1}(4)=b$. We have that $\alpha(G,\gamma)=2$ and is achieved by the sequence $\tau$.
}
\label{fig:seq-alphabet-size}
\end{figure}

%\knote{This figure is still not helpful to parsed the definition of $\alpha(G, \gamma)$. }
%\knote{Draw $G$ according to $\gamma$ order as the topmost figure. Next figure: mention $\tau$. Next figure draw $\permutationGraph(\tau)$. Draw one underneath the other.}

As a stepping stone towards understanding $\alpha(G)$, we show that $\alpha(G,\gamma)$ for a given topological ordering $\gamma$ of $G$ (i.e., the $\gamma$-alphabet size of $G$) can be computed in polynomial time. 

\begin{restatable}{theorem}{AlphabetSizewrtOrdering} \label{thm:alphabet-size-wrt-ordering-polytime-computability}
There exists a polynomial-time algorithm that takes a permutation DAG $G$ and a topological ordering $\gamma$ of $G$ as input and detects if $\alpha(G, \gamma)$ is finite and if so, then returns a sequence that achieves $\alpha(G, \gamma)$. 
\end{restatable}

Our algorithm underlying Theorem \ref{thm:alphabet-size-wrt-ordering-polytime-computability} also reveals a min-max relation for $\gamma$-alphabet size that we describe now. Let $G=(V, A)$ be a permutation DAG with $n$ vertices and let $\gamma:V\rightarrow [n]$ be a topological ordering of $G$ such that $\alpha(G, \gamma)$ is finite. Let $\overrightarrow{E} := \{(u,v): \gamma(u) < \gamma(v) \text{ and } (v, u) \not \in A\}$ and $H(G, \gamma):=(V, A\cup \overrightarrow{E})$. We note that $H(G, \gamma)$ is a tournament.\footnote{A tournament is a directed graph $H=(V,A)$ in which we have exactly one of the two arcs $(v,u)$ and $(u,v)$ for every pair of vertices $u,v\in V$.} Also, let $w:A\cup \overrightarrow{E} \rightarrow\{0,1\}$ be an arc weight function for $H(G, \gamma)$ defined as follows:
$$w(e) := \begin{cases}
    0& \text{ if } e \in A,\\
    1&  \text{ if } e \in \overrightarrow{E}.
\end{cases}$$ 
Then, we have the following min-max relation for the minimization problem corresponding to $\alpha(G,\gamma)$.

\begin{restatable}{theorem}{AlphabetSizeMinMax} \label{thm:alphabet-size-min-max-relation}
Let $G = (V,A)$ be a permutation DAG and $\gamma$ be a topological ordering of $V$ such that $\alpha(G, \gamma)$ is finite. 
Then, 
$$\alpha(G, \gamma) = 1 + \max \left\{\sum_{e\in P}w(e):P \text{ is a path in } H(G,\gamma)\right\}.$$
\end{restatable}

In addition to the algorithm and the min-max relation, we give a polyhedral result (see Theorem \ref{thm:polyhedral-description} in Section \ref{subsec:polyhedral_description}) that also leads to an LP-based algorithm to compute $\alpha(G, \gamma)$. 
We believe that alphabet size, as a parameter, is likely to be useful in the context of permutation DAGs and consequently, merits a thorough study. We view Theorems \ref{thm:alphabet-size-wrt-ordering-polytime-computability} and \ref{thm:alphabet-size-min-max-relation} as stepping stones towards the problem of efficiently computing the alphabet size of a given permutation DAG and Theorem \ref{thm:MBT-for-small-alphabet-size} to be an application of this parameter. Resolving the complexity of computing the alphabet size is an intriguing open problem. 

\medskip
Next, we address the maximum binary tree problem in undirected graphs with bounded treewidth. Here, we are given an undirected graph $G$ and the goal is to find a subgraph that is a binary tree with maximum number of nodes. An undirected graph is said to be a binary tree if the graph is acyclic and every vertex has degree at most $3$. We observe that the existence of a binary tree can be expressed as a monadic second order logic property and hence, extensions of Courcelle's theorem \cite{downey2012parameterized} can be used to obtain an algorithm for maximum-sized binary tree that runs in time $f(w)n$ for $n$-vertex undirected graphs with treewidth $w$ for some function $f(w)$. However, the run-time dependence $f(w)$ is at least doubly exponential on the treewidth $w$ in this approach. We improve this dependence substantially. 

\begin{theorem} \label{theorem:bounded-treewidth}
Given a tree decomposition of an $n$-vertex undirected graph $G$ with treewidth $w$, there exists an algorithm to find a maximum-sized binary tree in $G$ in time $w^{O(w)}n$. 
\end{theorem}

\subsection{Related Work}

Heapability of integer sequences was introduced in \cite{byers2010heapable} and has been investigated further 
in \cite{istrate2015heapable,porfilio2015combinatorial,istrate2016heapability,basdevant2016hammersley,basdevant2017almost,hammersley-fl,balogh2020computing}. Heapability of integer sequences can be decided by a simple  greedy algorithm \cite{byers2010heapable} (see also \cite{istrate2016heapability} for an alternate approach based on integer programming, and \cite{balogh2020computing} for connections with Dilworth's theorem and an algorithm based on network flows). Besides introducing the longest heapable subsequence problem, \cite{byers2010heapable} also showed that deciding if a  
sequence can be arranged in a \textit{complete} binary heap is NP-complete. 

Heapable sequences of integers can be regarded as ``loosely increasing''. The celebrated \emph{Ulam-Hammesley problem} aims to understand the length of the longest increasing sequence of a random permutation. This has a long history with deep connections to many areas of science (e.g., see \cite{romik2015surprising}). \cite{byers2010heapable} studied the counterpart of this problem for heapability: they showed that the longest heapable subsequence of a random permutatoin of length $n$ is of size $n-o(n)$ with high probability and it can also be found in an online fashion. 

As mentioned earlier, Porfilio \cite{porfilio2015combinatorial} showed that the longest heapable subsequence is equivalent to solving the maximum-sized binary tree problem in permutation DAGs. In an earlier work \cite{chandrasekaran2019maximum}, we showed that the maximum-sized binary tree problem is NP-hard in DAGs and showed further inapproximability results. We also gave a fixed-parameter algorithm for the maximum binary tree problem when parameterized by the solution size. Furthermore, we designed a polynomial-time algorithm to solve the maximum-sized binary tree problem in the special class of bipartite permutation graphs. It is also known that maximum-sized binary tree problem in DAGs induced by sets of intervals can be solved in polynomial time \cite{balogh2020computing}.

\paragraph{Organization.}  
In Section \ref{sec:DPforalphabetsize}, we present the fixed-parameter algorithm for longest heapable subsequence when parameterized by the alphabet size and prove Theorem \ref{thm:MBT-for-small-alphabet-size}. In Section \ref{sec:computingalphabetsize}, we address the problem of computing $\gamma$-alphabet size, present a min-max relation, and a polyhedral description for the same. In Section \ref{sec:boundedtreewidth}, we present a fixed-parameter algorithm for computing maximum-sized binary tree in a given undirected graph when parameterized by treewidth.
\newcommand{\refine}{\textup{refine}}
\newcommand{\LHS}{\textsf{LHS}}

\section{LHS for small alphabet sizes}\label{sec:DPforalphabetsize}

In this section we prove Theorem~\ref{thm:MBT-for-small-alphabet-size}. 

\mbtAlphabetK*

Before explaining our algorithm, we first establish certain useful definitions. 
We will denote a directed binary tree where each node is labeled by some natural number in $[k]$ such that the labels on every leaf to root path is non-increasing as a \emph{heap over alphabet $[k]$}. 

\begin{definition}[Extended Binary Tree]
Given a rooted non-empty binary tree $T$, we define the extended binary tree $Ext(T)$ by introducing new leaf nodes in a way that makes every node in $T$ have exactly 2 children. The nodes in $T$ are also referred to as \emph{internal nodes}, and the new leaf nodes are referred to as \emph{external nodes}.
\end{definition}

\begin{definition}[Shape]
Given a heap $H$ over alphabet $[k]$, we define its \emph{shape} as a tuple $\*x=\tp{x_0, x_1, \ldots, x_{k-1}, x_k}$, where $x_i$ is the number of external nodes whose parents have label $i$ in $Ext(H)$. We also follow the convention that the shape of an empty heap is $\tp{1, 0, \ldots, 0}$.
\end{definition}

\begin{figure}[ht]
\begin{center}
		\begin{tikzpicture}
		\node[vertex] (a1) at (1,1) {$3$}; 
		\node[vertex] (a2) at (2,2) {$1$};
		\node[rectangle, draw=black, fill=white, minimum size=15pt]  (a3) at (3,1) {$c$};
		\node[rectangle, draw=black, fill=white, minimum size=15pt] (a4) at (0,0) {$a$}; 
		\node[rectangle, draw=black, fill=white, minimum size=15pt]  (a5) at (2,0) {$b$};
		\path [->]
		(a1) edge (a2)
		(a3) edge (a2)
		(a4) edge (a1)
		(a5) edge (a1);
		\end{tikzpicture}
	\end{center}
\caption{Given a binary tree $T$ composed of two nodes $1$ and $3$, the extended binary tree $Ext(T)$ has two internal nodes $1$ and $3$, and three new external nodes $a$, $b$, and $c$. Suppose $k=4$. The shape of the heap above is $\*x=(x_0,x_1,x_2,x_3,x_4)=(0,1,0,2,0)$ because the parent of $a$ and $b$ has label 3 and the parent of $c$ has label 1.}
\end{figure}
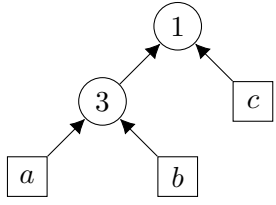

Intuitively, an external node represents the location of a potential future insertion into the heap. Since an insertion is effectively replacing an external node with a new internal node (thus introducing two new external nodes), it is captured by simple manipulations of shapes. This naturally leads us to defining insertions with respect to shapes. Given a shape $\*x=(x_0, \ldots, x_k)$ and labels $a \le b$, the shape obtained by inserting $b$ under $a$, denoted $\*x(a \gets b)$, is defined as 
\begin{align*}
\*x\tp{a \gets b}:=
\begin{cases}
\tp{x_0, \ldots, x_{a-1}, x_a+1, x_{a+1}, \ldots, x_k} & \text{if $x_a > 0$ and $a=b$,} \\
\tp{x_0, \ldots, x_{a-1}, x_a-1, x_{a+1}, \ldots, x_{b-1}, x_b+2, x_{b+1}, \ldots, x_k} & \text{if $x_a > 0$ and $a < b$,} \\
\perp & \text{if $x_a = 0$.}
\end{cases}
\end{align*}
For example, consider the shape $\*x = (0,1,0,2,0)$. The shape $\*x\tp{1 \gets 2}$ is $(0,0,2,2,0)$, and the shape $\*x(2 \gets 3)$ is $\perp$. Given a heap $H$ and a sequence $a=(a_1, \ldots, a_n)$, consider a longest subsequence of $a$ which can be sequentially inserted to $H$ as leaf nodes while maintaining the heap property. We will call such a subsequence a \emph{longest heapable subsequence starting from $H$}. We observe that the longest heapable subsequence of a given sequence starting from $H$ depends only the shape of $H$ and not the precise structure of $H$ (i.e., the optimum does not change for two different heaps $H_1$ and $H_2$ sharing the same shape). Therefore, it is equivalent and also convenient to consider the longest heapable subsequence problem starting from an initial shape instead of an initial heap. This line of thought also 
suggests a natural dynamic programming approach where the subproblems are specified by shapes. 

To analyze the running time, we need to upper bound the number of subproblems, which is the same as the number of distinct shapes. As a starting point, the number of distinct shapes can be upper bounded by $n^{O(k)}$. This is because in any $n$-node heap $H$ there are exactly $n+1$ external nodes in $Ext(H)$ (an elementary property of binary trees). Therefore the number of shapes is bounded by the number of non-negative integral solutions to $x_0+x_1+\ldots+x_{k}=n+1$, which is $n^{O(k)}$. Although this estimate seems like a very crude upper bound, bringing down the estimate into the fixed-parameter regime (i.e., $f(k)n^{O(1)}$) seems very difficult. We employ additional ideas to design a fixed-parameter algorithm.

Consider the longest heapable subsequence problem starting from initial shape $\*x=(x_0, x_1, \ldots, x_k)$. Suppose that the initial shape also satisfies the condition that $x_j\ge k-j+1$ for some $j\in [k]$. Our key observation is that all elements with labels at least $j$ are heapable from $\*x$: we can reserve an external node attached to $j$ for each label $v\in \{j, j+1, \ldots, k\}$, which can then be used to form a chain of elements with the same label $v$. Essentially, once we have reached the shape $\*x$, there are ``infinitely'' many external nodes available for future elements with label at least $j$, and hence, we no longer need to keep track of the precise values of $x_j, x_{j+1},\ldots, x_k $. This motivates the following notion of refined shapes.

\begin{definition}[Refined shapes]
A tuple $(x_0, x_1, \ldots, x_k)$ is a \emph{refined shape} (over alphabet size $k$) if for each $j \in \set{0, 1, \ldots, k}$ we have $x_j \in \set{0, 1, \ldots, k-j} \cup \set{\infty}$, and $x_j = \infty$ implies $x_\ell = \infty$ for all $\ell > j$. We will write $\+X_k$ for the set of all refined shapes over alphabet size $k$.
\end{definition}

The operation $\refine(\cdot)$ introduced below formalizes the intuition discussed earlier.
\begin{definition}
Let $\*x=(x_0, \ldots, x_k)$ be such that $x_j \in \mathbb{N} \cup \set{\infty}$ for all $j$. Let 
\begin{align*}
\refine(\*x) := \begin{cases}
\*x & \textup{if $x_j \le k-j$ for all $j$} \\
(x_0, \ldots, x_{j_0-1}, \infty, \infty, \ldots) & \textup{$j_0$ is the smallest $j$ such that $x_j \ge k-j+1$}
\end{cases}.
\end{align*}
\end{definition}

We remark that $\refine(\*x) \in \+X_k$ for any $\*x$. Next we define insertions with respect to refined shapes. Given a refined shape $\*x=(x_0, \ldots, x_k)$ and labels $a \le b$, the shape obtained by inserting $b$ under $a$, denoted $\*x(a \gets b)$, is defined as 

\begin{align*}
\*x\tp{a \gets b}:=\begin{cases}
\refine\tp{x_0, \ldots, x_{a-1}, x_a+1, x_{a+1}, \ldots, x_{k-1}} & \textup{if $x_a > 0$ and $a=b$,} \\
\refine\tp{x_0, \ldots x_{a-1}, x_a-1,x_{a+1}, \ldots, x_{b-1}, x_b+2, x_{b+1},\ldots, x_{k-1}} & \textup{if $x_a > 0$ and $a < b$,}\\
\perp & \textup{if $x_a = 0$.}
\end{cases}
\end{align*}
where we followed the convention that $\infty > 0$ and $\infty + c = \infty$ for any constant $c$. 

Now we are ready to state the dynamic programming algorithm. In the following, we fix $\tp{a_1, a_2, \ldots, a_n}$ as the input sequence. For $\*x \in \+X_k$ and $i \in [n]$ define $\LHS[i, \*x]$ to be the length of the longest heapable subsequence in the prefix sequence $\tp{a_1, a_2, \ldots, a_i}$, with an additional constraint that the refined shape of the heap constructed from the subsequence should be $\*x$. We write $\LHS[i, \*x] = -\infty$ if there is no feasible solution (i.e. shape $\*x$ is not reachable by any subsequence of $\tp{a_1, \ldots, a_i}$). With this definition, the longest heapable subsequence of the given sequence has length $\max_{\*x\in \+X_k}\LHS[n,\*x]$. Our goal now is to compute $\LHS[n, \*x]$ for each $\*x \in \+X_k$. 

For a label $v \in \set{1, 2, \ldots, k}$ and two refined shapes $\*x$ and $\*x'$, we say that \emph{$\*x$ is reachable from $\*x'$ via an insertion of $v$} if there exists $b \le v$ such that $\*x'(b \gets v) = \*x$. We denote by $\texttt{prev}(\*x, v)$ the set of refined shapes from which $\*x$ is reachable via an insertion of $v$. We show that $\LHS$ satisfies the following recurrence relation. 

\begin{lemma} \label{lem:recurrence-relation}
For every $i\in [n]$ and $\*x \in \+X_k$, we have that 
\[
\quad \LHS[i, \*x] = \max\set{\LHS[i-1, \*x], \max_{\*x' \in \texttt{prev}(\*x, a_i)} \set{\LHS[i-1, \*x']}+1 }.
\]
\end{lemma}
\begin{proof}
We will show that $\LHS[i, \*x] \le \max\set{\LHS[i-1, \*x], \max_{\*x' \in \texttt{prev}(\*x, a_i)} \set{\LHS[i-1, \*x']}+1 }$ as the other direction is trivial. Let us fix an optimal heapable subsequence $s$ of $\tp{a_1, \ldots, a_i}$. If $a_i$ does not belong to $s$, it must be the case that $s$ is also an optimal heapable subsequence of $\tp{a_1, \ldots, a_{i-1}}$. In this case $\LHS[i, \*x] = \LHS[i-1,\*x]$. If $a_i$ belongs to $s$, we further fix an optimal heap $H$ (with refined shape $\*x$) and assume that $a_i$ is inserted under an element with value $b$ in $H$. Removing $a_i$ from $H$ results in a heap $H'$ with a shape $\*x'$ satisfying $\*x'(b \gets a_i) = \*x$. In particular, $\*x' \in \texttt{prev}(\*x, a_i)$. In this case, $\LHS[i,\*x] = \LHS[i-1,\*x'] + 1 \le \max_{\*x' \in \texttt{prev}(\*x, a_i)} \set{\LHS[i-1, \*x']}+1$.
\end{proof}

\begin{proof}[Proof of Theorem~\ref{thm:MBT-for-small-alphabet-size}]
Given Lemma~\ref{lem:recurrence-relation}, it remains to show that the recurrence relation can be implemented in time $(k+1)!\cdot k \cdot O(n)$. We observe that the number of subproblems is bounded by $O(n\abs{\+X_k})$. The set $\texttt{prev}(\*x, a_i)$ can be enumerated in time $O(k)$ by inverting the operation $\*x'(b \gets a_i)$ for each $b \le a_i$. Therefore it suffices to show that $|\+X_k| = O((k+1)!)$.

In order to bound the size of $\+X_k$, we observe that for every $\*x = (x_0, x_1, \ldots, x_k) \in \+X_k$, we have that $x_0 = 0$ unless $\*x=(1,0,\ldots,0)$, and that $x_j \in \set{0, 1, \ldots, k-j} \cup \set{\infty}$ for $j \ge 1$. Therefore $|\+X_k| \le 1+ \prod_{j=1}^{k}(k-j+2) = (k+1)! + 1$. 
\end{proof}

An implementation of this dynamic programming algorithm is given by Algorithm~\ref{alg:LHS-alphabet-k}. This implementation requires space complexity $O((k+1)!n\cdot \log n)$, which can be optimized to $O((k+1)! \cdot \log n)$ using a standard rolling array technique: we observe that in the recurrence relation, $\LHS[i, \*x]$ depends only on $\LHS[i-1, \*x']$ but not on $\LHS[j, \*x']$ for any $j < i-1$. Therefore the values $\LHS[i-2, \*x]$ become obsolete and the space can be recycled to store new values. Essentially, we only need two arrays $\LHS_1[\*x]$ and $\LHS_2[\*x]$ and store new values alternately between them. 

\medskip
\noindent \textbf{Remark.} We note that our dynamic programming algorithm also works in the streaming model, where the elements of the input sequence have to be processed one by one without storing them in memory and the goal is to find the length of a longest heapable subsequence of the input that has arrived so far. For constant alphabet size $k$, the space complexity of our algorithm is $O\tp{\log n}$.

\begin{algorithm}
	\caption{Longest Heapable Subsequence for Alphabet Size $k$} \label{alg:LHS-alphabet-k}
	\textbf{Input:} A sequence $a=\tp{a_1, \ldots, a_n}$ such that $\forall i \in [n]$, $a_i \in \set{1, 2, \ldots, k}$. \\
	\textbf{Output:} The length of longest heapable subsequence in $a$. \\
	
	$\mathbf{\texttt{LHS}}(a_1, a_2, \ldots, a_n):$
	\begin{algorithmic}[1]
	    \State $\+X \gets \set{(1,0,\ldots,0)}$ \Comment{$\+X$ maintains a set of reachable refined shapes}
		\State $\LHS \gets$ integer array of size $n \times (k+1) \times k \times \ldots 2 \times 1$
		\State $\LHS[0, (1,0, \ldots, 0)] \gets 0$
		\For {$i \gets 1$ to $n$} \Comment DP main body
		    \For{$\*x \in \+X$}
		        \State $\LHS[i, \*x] \gets dp[i-1, \*x]$ \Comment{Discard $a_i$}
		    \EndFor
			\For{$\*x \in \+X$}
			    \For{$b \in \set{b' \colon 0 \le b' \le a_i, x_{b'} > 0}$ }
			        \State $\*x' \gets \*x(b \gets a_i)$ \Comment{Insert $a_i$ under $b$ to reach refined shape $\*x'$}
			        \If {$\*x' \notin \+X$} \Comment{First time reaching shape $\*x'$}
			            \State $\+X \gets \+X \cup \set{\*x'}$
			            \State $\LHS[i, \*x'] \gets \LHS[i-1, \*x] + 1$
			        \ElsIf{$\LHS[i, \*x'] < \LHS[i-1, \*x] + 1$} 
			            \State $\LHS[i, \*x'] \gets \LHS[i-1, \*x] + 1$
			        \EndIf
			    \EndFor
			\EndFor
		\EndFor
		\Return{$\max\set{\LHS[n, \*x] \colon \*x \in \+X}$}
	\end{algorithmic}
\end{algorithm}
\newcommand{\LFSC}{\ensuremath{\mathsf{LFSC}}\xspace}
\newcommand{\gammaLFSC}{\ensuremath{\gamma\text{-}\mathsf{LFSC}}\xspace}
\newcommand{\orderProjection}{\ensuremath{\mathsf{Proj}}\xspace}
\section{Alphabet Size of Permutation DAGs}\label{sec:computingalphabetsize}
In this section, we consider the problem of computing the $\gamma$-alphabet size of a permutation DAG where $\gamma$ is a given topological ordering of $G$. 
We give an efficient algorithm 
in Section \ref{subsec:Alphabet-Size-Computation}, a min-max relation in Section \ref{subsec:alphabet-size-wrt-ordering-minmax-relation}, and a polyhedral description in Section \ref{subsec:polyhedral_description}. 

We recall that a directed graph $G$ is a permutation DAG if there exists a sequence $\sigma$ such that $\permutationGraph(\sigma)$ is isomorphic to $G$. We note that permutation DAGs are \emph{transitively closed}, i.e., for a permutation DAG $G=(V,A)$, if $(u, v), (v, w) \in A$, then $(u, w) \in A$. In order to recognize if a given DAG is a permutation DAG, we need the notion of \emph{umbrella-free ordering} defined below (see Figure \ref{fig:umbrella-free} for an example). This notion will also help us recognize if $\alpha(G,\gamma)$ is finite. 

\begin{definition}[Umbrella-free Order]
Let $G = (V, A)$ be an $n$-vertex DAG. An order $\gamma:V \rightarrow [n]$ of $V$ is \emph{umbrella-free} if for all $(v,u) \in A$ and for every vertex $w \in V$ with $\gamma(u) < \gamma(w) < \gamma(v)$, either $(w, u) \in A$ or $(v, w) \in A$ (or both).
\end{definition}

\begin{figure}[ht]
\begin{center}
\begin{tikzpicture}[scale =0.7]
\coordinate (O) at (0,0); 

%% Umbrella
\node [draw, circle] at (4, 0)(1) {$u$};
\node [draw, circle] at (6, 0) (2) {$w$};
\node [draw, circle] at (8, 0) (3) {$v$};
\draw[->, bend right] (3) to (1);

%% Ordering with Umbrella
%\node [draw, circle, scale=0.6] at (8, 0)(1) {};
\node [draw, gray, circle] at (12, 0) (a) {$a$};
\node [draw, circle, thick] at (14, 0) (b) {$b$};
\node [draw, circle, thick] at (16, 0) (c) {$c$};
\node [draw, circle, thick] at (18, 0) (d) {$d$};
\draw[->, bend right, gray] (b) to (a);
\draw[->, bend right, gray] (c) to (a);
\draw[->, bend right, gray] (d) to (a);
\draw[->, bend right, thick] (d) to (b);

\node  at (6, -1) () {\small(a)};

\node  at (19, -1) () {\small(b)};

%% Umbrella-Free Ordering
\node [draw, circle] at (20, 0) (A) {a};
\node [draw, circle] at (22, 0) (B) {b};
\node [draw, circle] at (24, 0) (D) {d};
\node [draw, circle] at (26, 0) (C) {c};

\draw[->, bend right] (B) to (A);
\draw[->, bend right] (C) to (A);
\draw[->, bend right] (D) to (A);
\draw[->, bend right] (D) to (B);
\end{tikzpicture}
\end{center}
\caption{(a) Depicts the scenario when the triple $(u,w,v)$ is an \emph{umbrella}. (b) Shows two topological orderings of the same DAG. The order $(a,b,c,d)$ is not umbrella-free due to the (highlighted) umbrella $(b,c,d)$, while the the order $(a,b,d,c)$ is umbrella-free.}
\label{fig:umbrella-free}
\end{figure}
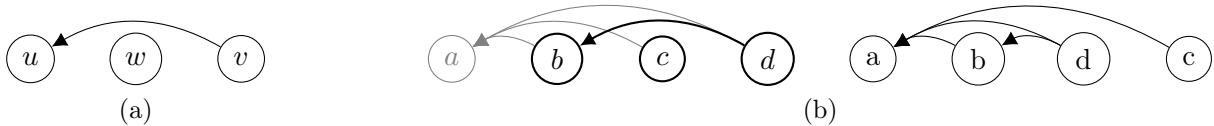

The following lemma characterizes permutation DAGs in terms of the 
existence of an umbrella-free topological ordering. 

\begin{lemma}[\cite{pnueli_lempel_even_1971, GOLUMBIC1980157}]\label{lem:permutation-dag-iff-umbrella-free-ordering}
Let $G=(V,A)$ be a transitively closed DAG. Then $G$ is a permutation DAG if and only if there exists an umbrella-free topological ordering of $G$. Moreover, there exists a polynomial-time algorithm to verify if a given DAG is a permutation DAG and if so, then construct an umbrella-free topological ordering. 
\end{lemma}

Lemma \ref{lem:permutation-dag-iff-umbrella-free-ordering} implies that $\alpha(G,\gamma)$ is finite if and only if $\gamma$ is an umbrella-free topological ordering of $G$.

\subsection{Computing $\gamma$-alphabet size}\label{subsec:Alphabet-Size-Computation}
In this section, we address the problem of computing the $\gamma$-alphabet size of a given permutation DAG, where $\gamma$ is a topological ordering of $G$. 
The following is the main result of this section.

\AlphabetSizewrtOrdering*

We note that umbrella-freeness of a given topological ordering can be verified in polynomial-time, so we may henceforth assume that the input $\gamma$ is in fact an umbrella-free topological ordering of $G$. 
We will give an iterative algorithm to compute $\alpha(G, \gamma)$. We observe that computing $\alpha(G,\gamma)$ involves assigning a value to each vertex of $G$ such that the sequence obtained by 
ordering the values of the vertices in the same order as $\gamma$ gives the same permutation DAG as $G$. 
At each iteration, our algorithm will choose a vertex of $G$ and assign a value to it. The next definition will allow us to formally define the choice of this vertex.

\begin{definition}[Fully Suffix Connected Vertex]
Let $G = (V, A)$ be a permutation DAG and $\gamma$ be a topological ordering of $G$. 
A vertex $u \in V$ is \emph{fully suffix connected} if for all $v \in V$ such that $\gamma (v) > \gamma(u)$, we have $(v,u) \in A$. The \emph{$\gamma$-least fully suffix connected ($\gammaLFSC$) vertex} is the fully suffix connected vertex $u$ with smallest $\gamma(u)$.
\end{definition}

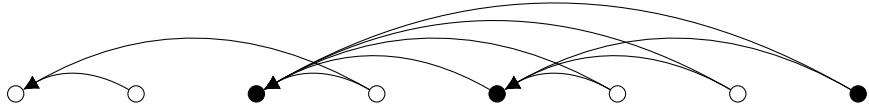
\begin{figure}[ht]
\begin{center}
\begin{tikzpicture}[scale =0.8]
\coordinate (O) at (0,0); 

%% Fully Suffix Connected
\node [draw, circle, scale=0.6] at (2, 0)  (1) {};
\node [draw, circle, scale=0.6] at (4, 0)  (2) {};
\node [draw, circle, scale=0.6, fill] at (6, 0)  (3) {};
\node [draw, circle, scale=0.6] at (8, 0)  (4) {};
\node [draw, circle, scale=0.6, fill] at (10, 0) (5) {};
\node [draw, circle, scale=0.6] at (12, 0) (6) {};
\node [draw, circle, scale=0.6] at (14, 0) (7) {};
\node [draw, circle, scale=0.6, fill] at (16, 0) (8) {};

%\draw[->, bend right] (3) to (2);
\draw[->, bend right] (4) to (3);
\draw[->, bend right] (5) to (3);
\draw[->, bend right] (6) to (3);
\draw[->, bend right] (7) to (3);
\draw[->, bend right] (8) to (3);
\draw[->, bend right] (6) to (5);
\draw[->, bend right] (7) to (5);
\draw[->, bend right] (8) to (5);
%\draw[->, bend right] (6) to (3);
%\draw[->, bend right] (7) to (1);
\draw[->, bend right] (2) to (1);
\draw[->, bend right] (4) to (1);

\end{tikzpicture}
\end{center}

\caption{The DAG in the given topological order $\gamma$ has $3$ fully suffix connected vertices that are depicted as filled circles. The leftmost fully suffix connected vertex is the (unique) $\gammaLFSC$ vertex.}
\label{fig:Fully_Suffix_Connected}
\end{figure}
See Figure \ref{fig:Fully_Suffix_Connected} for an example showing fully suffix connected vertices. We note that \gammaLFSC is unique. The following lemma states a useful property of the \gammaLFSC vertex.

\begin{lemma}\label{lem:gammaLFSC-vertex-no-outgoing-arcs}
Let $G = (V,A)$ be a permutation DAG and $\gamma$ be an umbrella-free topological ordering of $G$. Then, the \gammaLFSC vertex has no outgoing arcs in $G$.
\end{lemma}
\begin{proof}
 Let $v \in V$ be the \gammaLFSC and suppose for contradiction that $v$ has an outgoing arc in $G$.
 Let $u$ be the vertex with largest $\gamma(u)$ such that $(v,u)\in A$. 
 We note that $\gamma(u)<\gamma(v)$ since $\gamma$ is a topological ordering. 
 We will show that such a vertex $u$ is fully suffix connected and hence contradicts the $\gamma$-least fully suffix connected property of vertex $v$.
 
 We first show that for every vertex $w\in V$ such that $\gamma(w)\ge \gamma(v)$, we have $(w, u)\in A$. 
 For $w=v$, this follows since $(v,u)\in A$ by the choice of $u$.
 Let $w$ be a vertex such that $\gamma(w)>\gamma(v)$. Since $v$ is fully suffix connected, we have that $(w,v)\in A$. Also, since $G$ is a permutation DAG, it is transitively closed. Hence, $(v,u)\in A$ implies that $(w,u)\in A$. 
 
 Next, we show that for every vertex $w\in V$ such that $\gamma(u)<\gamma(w)<\gamma(v)$, we have $(w, u)\in A$. Let $w$ be a vertex such that $\gamma(u)<\gamma(w)<\gamma(v)$. By assumption, the ordering $\gamma$ is umbrella-free. Thus, at least one of $(w, u)$ or $(v,w)$ must exist in $A$. However, $(v,w) \not \in A$ as otherwise, $w$ will contradict the choice of vertex $u$. Therefore, $(w, u)\in A$.
 \end{proof}

We now discuss a high level overview of our iterative greedy algorithm for computing $\alpha(G, \gamma)$. During the first iteration, 
the algorithm greedily chooses the \gammaLFSC vertex $v_1$ (say) in $G_1:=G$ to assign the smallest alphabet, namely $\sigma(v_1)=1$. 
The vertex $v_1$ and its incident edges are deleted from $G_1$ to form $G_2$, and the remaining 
$n-1$ vertices $V\backslash\{v\}$ are 
ordered in the same relative order as $\gamma$---denote this ordering as $\gamma_2$. 
In the second iteration, our algorithm  greedily chooses the $\gamma_2$-$\LFSC$ vertex $v_2$ (say) in $G_2$ to assign the next smallest alphabet---the next smallest alphabet is chosen based on whether $v_2$ lies to the left or right of $v_1$: if $v_2$ lies to the left of $v_1$ with respect to $\gamma$, then we set $\sigma(v_2)=\sigma(v_1)+1$, otherwise we set $\sigma(v_2)=\sigma(v_1)$. 
This iterative removal and assignment process continues for $n$ iterations, i.e., until all  vertices are removed from $G$. The final output sequence will just be the sequence of assigned values in the order of vertices in $\gamma$. Before presenting our final complete 
algorithm (Algorithm \ref{alg:greedyassign}), we introduce a definition to formalize the reordering of vertices after removing a vertex from $G$ --- this will allow us to obtain $\gamma_{i+1}$ from $\gamma_i$.

\begin{definition}[Projected order]
Let $G = (V,A)$ be an $n$-vertex 
DAG, and $\gamma$ be a topological ordering of $V$. Let $H = G - v$. Then the \emph{projection of $\gamma$ onto $H$}, denoted by $\orderProjection_H[\gamma]: V\backslash\{v\} \rightarrow [n-1]$, is defined as $$\orderProjection_H[\gamma](u) = \begin{cases}
    \gamma(u)& \text{ if } \gamma(u) < \gamma(v),\\
    \gamma(u) - 1& \text{ if } \gamma(u) > \gamma(v).
\end{cases}$$
\end{definition}

Armed with the notions of fully suffix connected vertices and projected order, we state our algorithm below.

\begin{algorithm}[H]
\caption{\GreedyAssign algorithm to compute $\alpha(G, \gamma)$}
\small
\label{alg:greedyassign}
\textbf{Input:} Permutation DAG $G = (V,A)$ on $n$ vertices in umbrella-free topological order $\gamma:V\rightarrow [n]$.\\
\textbf{Output:} Sequence $\sigma$ of length $n$. \\

$\mathbf{\GreedyAssign}(G, \gamma)$:
\begin{algorithmic}[1]
\State{Initialize $\alpha \gets 1;\ \ G_1 \gets G; \ \ \gamma(v_0) \gets -\infty;\ \ \gamma_1 \gets \gamma$}
\For{$i \gets 1 \textup{ to } n$}
    \State{$v_i \gets \gamma_i\text{-}\LFSC$ in $G_i$}
    \If{$\gamma(v_i) < \gamma(v_{i-1})$}{\ $\alpha \gets \alpha + 1$}
    \EndIf
    \State{$G_{i+1} \gets G_i - v_i$}
    \State{$\gamma_{i+1} \gets \orderProjection_{G_{i+1}}[\gamma_i]$}
    \State{$\sigma(v_i) \gets \alpha$}
\EndFor
\State{\textbf{Return} $\sigma \gets (\sigma(\gamma^{-1}(1)) \ldots \sigma(\gamma^{-1}(n)))$}
\end{algorithmic}
\end{algorithm}

The algorithm can be implemented to run in polynomial-time since a $\gamma$-LFSC vertex in $G$ can be computed in polynomial-time. We now prove the 
correctness of the algorithm.
Let $G = (V, A)$ be an $n$-vertex 
permutation DAG, and $\gamma$ be an umbrella-free topological ordering of $G$. Let $v_1,\ldots, v_n$ 
be the sequence of vertices chosen in the execution of
$\GreedyAssign(G, \gamma)$. Let $\alpha_i, G_{i}$ and $\gamma_i$ denote the alphabet size $\alpha$ at the end of the $i^{th}$ iteration, the remaining subgraph at the start of the $i^{th}$ iteration, and  $\gamma$ projected onto $G_i$ respectively. Finally, let $\sigma$ be the sequence returned by $\GreedyAssign(G, \gamma)$. We have the following observations 
about the execution of the algorithm.

\begin{observation}\label{obs:greedyAssign-removal-outdegree}
The vertex $v_i$ has no outgoing arcs in $G_i$ for all $i \in [n]$.
\end{observation}

\begin{observation}\label{obs:GreedyAssign-alphabets-assigned-during-adjacent-iterations}
If $\gamma(v_{i+1}) < \gamma(v_i)$ then $\sigma(v_{i+1}) = \sigma(v_i) + 1$ and $(v_i, v_{i+1}) \not \in A$,  otherwise $\sigma(v_{i+1}) = \sigma(v_i)$ and $(v_{i+1}, v_i) \in A$. Thus, alphabet assignments by \GreedyAssign are non-decreasing with increasing  iterations i.e. $\sigma(v_i) \leq \sigma(v_j)$ for all $i,j \in [n]$ with $i < j$.
\end{observation}

Observation \ref{obs:greedyAssign-removal-outdegree} directly follows from Lemma \ref{lem:gammaLFSC-vertex-no-outgoing-arcs}. Observation \ref{obs:GreedyAssign-alphabets-assigned-during-adjacent-iterations} is due to the conditional increment of the alphabet size, $\alpha$, in \GreedyAssign. 
The next two lemmas show feasibility and optimality of 
\GreedyAssign respectively. Theorem \ref{thm:alphabet-size-wrt-ordering-polytime-computability} then immediately follows from Lemmas \ref{lem:GreedyAssign-Feasability} and \ref{lem:GreedyAssign-Optimality}.

\begin{lemma}[Feasibility of \GreedyAssign]\label{lem:GreedyAssign-Feasability}
Let $\permutationGraph(\sigma) = (\{t_1\ldots t_n\}, A')$  
Then $\permutationGraph(\sigma)$ 
is isomorphic to $G$ under the mapping $\phi:\{u_1 \ldots u_n\} \rightarrow V$ given by $\phi(u_i) = \gamma^{-1}(i)$.
\end{lemma}
\begin{proof}
We will prove isomorphism of the two graphs under $\phi$ by showing that $(u,v) \in A$ if and only if $(\phi^{-1}(u), \phi^{-1}(v)) \in A'$.

For the forward direction, it suffices to show that $\sigma(u) \leq \sigma(v)$ whenever $(u,v) \in A$. 
We observe that if $(u,v) \in A$, then $\gamma(v) < \gamma(u)$. By Observation \ref{obs:greedyAssign-removal-outdegree}, \GreedyAssign must assign $\sigma(v)$ before $\sigma(u)$. Observation \ref{obs:GreedyAssign-alphabets-assigned-during-adjacent-iterations} then implies that $\sigma(u) \leq \sigma(v)$. 

Next we show the contrapositive of the converse direction. Assume that $(u,v) \not \in A$. We first consider 
the case when $\gamma(v) > \gamma(u)$. Let $\phi^{-1}(u) = t_{\gamma(u)}$ and  $\phi^{-1}(v) = t_{\gamma(v)}$. By definition of permutation DAGs, $\permutationGraph(\sigma)$ does not have arc $(t_i, t_j)$ when $i < j$. Thus $(t_{\gamma(u)}, t_{\gamma(v)}) \not \in A$. 
Next, we consider 
the case when $\gamma(v) < \gamma(u)$. 
For this, it suffices to show that $\sigma(u) < \sigma(v)$. Since $(u,v) \not \in A$, 
the vertex $v$ will never become fully suffix connected before the removal of $u$. Thus \GreedyAssign sets $\sigma(u)$ before $\sigma(v)$. Thus, by Observation \ref{obs:GreedyAssign-alphabets-assigned-during-adjacent-iterations}, we have that $\sigma(u) \leq \sigma(v)$. Let $u = v_i$ and $v = v_j$, where $i, j \in [n]$ are the iteration numbers during which \GreedyAssign assigns $\sigma(u)$ and $\sigma(v)$ respectively. Then, there exists $k$ 
such that $i \leq k < j$ and $\gamma(v_{k+1}) < \gamma(v_k)$ as otherwise, Observation \ref{obs:GreedyAssign-alphabets-assigned-during-adjacent-iterations} would imply that $\gamma(v_i) < \gamma(v_j)$, a contradiction. Thus, $\sigma(u) < \sigma(v)$.
\end{proof}

\begin{lemma}[Optimality of \GreedyAssign]\label{lem:GreedyAssign-Optimality}
Let $\optsequence$ be a sequence achieving $\alpha(G, \gamma)$. Then, $\sigma(v_i) \leq \optsequence(v_i)$ for all $i \in [n]$.
\end{lemma}
\begin{proof}
We will show this by induction on $i$. For the base case of 
$i = 1$, \GreedyAssign always sets $\sequence(v_1) = 1$, the smallest possible alphabet assignment. Thus $\sequence(v_1) \leq \optsequence(v_1)$ holds. For the induction step, let $i\ge 2$. 
We have the following two cases based on whether \GreedyAssign incremented the alphabet size while assigning $v_{i}$.
\begin{enumerate}
    \item %
    Suppose $\sigma(v_{i}) = \sigma(v_{i-1})$. 
    By the description of the algorithm \GreedyAssign, we have that $v_{i-1}$ is fully suffix connected in $G_{i-1}$ and $\gamma(v_{i}) > \gamma(v_{i-1})$. Thus, the arc $(v_{i}, v_{i-1})$ must exist in $G_{i-1}$ and so also in $G$. It follows that $$\sigma(v_{i}) = \sigma(v_{i-1}) \leq \optsequence(v_{i-1}) \leq \optsequence(v_{i}).$$
    Here, the first inequality is by the induction hypothesis, while the second inequality is due to the observation that $(v_{i}, v_{i-1}) \in A$.
    
    \item %\knote{Suppose} 
    Suppose $\sigma(v_{i}) \not= \sigma(v_{i-1})$. 
     By the description of the algorithm \GreedyAssign, we have that 
     $\gamma(v_{i}) < \gamma(v_{i-1})$. Thus by Observation \ref{obs:greedyAssign-removal-outdegree}, the arc $(v_{i-1}, v_{i})$ does not exist in $G_{i-1}$ and hence, does not exist in $G$. It follows that
     $$\sigma(v_{i}) = \sigma(v_{i-1}) +1 \leq \optsequence(v_{i-1}) +1 \leq \optsequence(v_{i}).$$
     The equality relation is due to Observation \ref{obs:GreedyAssign-alphabets-assigned-during-adjacent-iterations}. The first inequality is due to the induction hypothesis while the second inequality is due to our observation that $(v_{i-1}, v_{i}) \not \in A$.
\end{enumerate}
\end{proof}

\noindent \textbf{Remark.} Algorithm \ref{alg:greedyassign} can be implemented to run in $O(|V| + |A|)$ time. This can be done by using a \emph{priority queue} data structure initialized as a \emph{stack}. All fully suffix connected vertices should be added to the priority queue with priorities being position in $\gamma$. The choice of vertex to assign is the vertex with the \emph{minimum} priority. 
The alphabet size should be incremented whenever a vertex removal results in new vertices becoming fully suffix connected.

\subsection{Min-Max Relation}\label{subsec:alphabet-size-wrt-ordering-minmax-relation}
Min-max relations are significant in optimization literature as they are strong indicators for the existence of a polynomial-time algorithm. In the context of algorithm design for optimization problems, min-max relations bring the optimization problem into $\NP \cap \coNP$, thus providing strong evidence for the existence of polynomial-time algorithms. 
In this section, we give a min-max relation for 
$\alpha(G, \gamma)$. An immediate consequence of our min-max relation 
will be an alternative linear time algorithm for computing $\alpha(G,\gamma)$. We believe that the min-max relation could be a useful tool towards computing $\alpha(G)$.

Let $G = (V,A)$ be a permutation DAG and $\gamma$ be an umbrella-free 
topological ordering of its vertices. We recall that $\overrightarrow{E} := \{(u,v): \gamma(u) < \gamma(v) \text{ and } (v, u) \not \in A\}$ and $H(G,\gamma)$ is the graph $(V,A\cup \overrightarrow{E})$. We also recall the arc weight function $w:A\cup \overrightarrow{E} \rightarrow\{0,1\}$ given by:
$$w(e) := \begin{cases}
    0& \text{ if } e \in A,\\
    1&  \text{ if } e \in \overrightarrow{E}.
\end{cases}$$ 
We observe that $H(G,\gamma)$ is a tournament. 
We now restate and prove the min-max relation. Our proof of the min-max relation will rely on our algorithm from Section \ref{subsec:Alphabet-Size-Computation}.
\AlphabetSizeMinMax*
\begin{proof}
We will show the  equation by showing inequality in both directions. We begin by showing the lower bound on $\alpha(G, \gamma)$. 
Let $P$ be any path in $H(G, \gamma)$, and $\sigma$ be any sequence such that $\permutationGraph(\sigma) = (\{t_1, \ldots, t_n\}, A')$ is isomorphic to $G$ under the mapping $\phi:\{t_1, \ldots, t_n \} \rightarrow V$ given by $\phi(t_i) = \gamma^{-1}(i)$. 
For every arc $(\phi(t_i), \phi(t_j)) \in P$ such that $(\phi(t_i), \phi(t_j)) \in \overrightarrow{E}$, we have the following two observations. First, the arc $(\phi^{-1}(t_j),\phi^{-1}(t_i)) \not \in A'$ as the arc $(\phi(t_j), \phi(t_i)) \not\in A$. Second,  $\sigma(\phi(t_i)) \geq \sigma(\phi(t_j)) + 1$ as $i < j$. It follows that 
$$ w(P) = \sum_{(u,v) \in P}w(u,v) = \sum_{(u,v) \in P \cap \overrightarrow{E}}w(u,v) \leq   \sum_{(u,v) \in P\cap \overrightarrow{E}}\sigma(u) - \sigma(v) \leq \alpha(G,\gamma) - 1.$$
The first and second equations are by definition of $w(P)$ and the weight function $w$ respectively. The first inequality is due to our observation that $\sigma(u) \geq \sigma(v) + 1$ whenever $(u,v) \in \overrightarrow{E}$. Let $a$ and $b$ be the first and last vertices on $P$. Then the final inequality follows from $\sigma(a) \geq 1$ and $\sigma(b) \leq \alpha(G, \gamma)$.

Next, we show the upper bound on $\alpha(G, \gamma)$. 
We recall that $v_1, \ldots, v_n$ is the order in which \GreedyAssign processes vertices of $G$. Consider $P = (v_n, v_{n-1}, \ldots, v_1)$. To prove the upper bound, it suffices to show that $(1)$ $P$ is a path in $H(G, \gamma)$; 
and $(2)\ w(P) \geq \alpha(G,\gamma) - 1$. To prove $(1)$, we show that $(v_i, v_{i-1}) \in A \cup \overrightarrow{E}$ for each $i \geq 2$. Consider the case when $\gamma(v_i) > \gamma(v_{i-1})$. Since $v_{i-1}$ was $\gamma_{i-1}\text{-}\LFSC$ in $G_i$, the arc $(v_i, v_{i-1}) \in A$. Next, consider the case when $\gamma(v_i) < \gamma(v_{i-1})$. By Observation \ref{obs:greedyAssign-removal-outdegree}, we have that the arc $(v_{i-1}, v_i) \not \in A$. Thus, the arc $(v_i, v_{i-1}) \in \overrightarrow{E}$ by definition of $\overrightarrow{E}$.
We now prove (2). 
By Observation \ref{obs:GreedyAssign-alphabets-assigned-during-adjacent-iterations}, and definitions of $w$ and $\overrightarrow{E}$, we have $w(v_i, v_{i-1}) = \sigma(v_i) - \sigma(v_{i-1})$ It follows that
$$ w(P) = \sum_{i = 2}^n w(v_i, v_{i-1})  = \sum_{i = 2}^n \sigma(v_i) - \sigma(v_{i-1}) = \alpha(G,\gamma) - 1.$$
The second equality is due to our previous observation. The final equality is due to the \GreedyAssign assignments $\sigma(v_n) = \alpha(G, \gamma)$ and $\sigma(v_1) = 1$.
\end{proof}

We remark that although the RHS problem in the min-max relation given in Theorem \ref{thm:alphabet-size-min-max-relation} is the longest path problem in a directed graph, it can be solved in the graph $H(G, \gamma)$ owing to the following lemma.
Lemma \ref{lem:H(G,gamma)_DAG} allows the optimization problem in the RHS of Theorem \ref{thm:alphabet-size-min-max-relation} to be solved 
in $O(|V| + |A|)$ time by the classical dynamic programming algorithm for maximum weight path in a DAG. This leads to 
an alternative algorithm for computing $\alpha(G, \gamma)$.

\begin{lemma}\label{lem:H(G,gamma)_DAG}
$H(G, \gamma)$ is a DAG.
\end{lemma}
\begin{proof}
Suppose for contradiction that $H(G, \gamma)$ contains a cycle. Let $C = (u_1, u_2, \ldots u_k, u_1)$ 
be a cycle with the smallest number of vertices. If $(u_1, u_3) \in A \cup \overrightarrow{E}$, then $C' = (u_1, u_3, \ldots u_k, u_1)$ is a cycle, contradicting our choice of $C$. Since $H(G, \gamma)$ is a tournament, the arc $(v_3, v_1) \in A \cup \overrightarrow{E}$, and $C' = (u_1, u_2, u_3, u_1)$ is also a cycle i.e. $k=3$. Recall that the subgraph $(V, A)$ is transitively closed. Thus, at most one edge of $C$ can belong to $A$. To get the required contradiction, it suffices to show that the subgraph $(V, \overrightarrow{E})$ is transitively closed. 
Suppose for contradiction that $\overrightarrow{E}$ is not transitively closed. Then, there exist arcs $(u,v), (v,w) \in \overrightarrow{E}$ such that the arc $(u,w) \not\in \overrightarrow{E}$. By definition of $\overrightarrow{E}$, we have that $\gamma(u) < \gamma(v) < \gamma(w)$. It follows that the arc $(w, u) \in A$, and the triple $(u,v,w)$ is an umbrella in $G$ ordered by $\gamma$. This contradicts that $\gamma$ is umbrella-free.
\end{proof}

\subsection{Polyhedral Description}\label{subsec:polyhedral_description}

In this section, we give a polyhedral description for the convex-hull of sequences that are feasible for $\alpha(G, \gamma)$. As a consequence, it leads to an LP-based algorithm to compute $\alpha(G, \gamma)$. We emphasize that our polyhedral result is stronger than giving an LP-based algorithm to compute $\alpha(G, \gamma)$: it implies that one can efficiently compute an \emph{integer-valued} sequence $\sigma=(\sigma(1), \ldots, \sigma(n))$ with minimum weight $\sum_{i=1}^n w_i \sigma(i)$ for any given non-negative weights $w_1, \ldots, w_n$ such that $G(\sigma)$ is isomorphic to $G$ under the mapping $\phi: \{t_1, \ldots, t_n\}\rightarrow V$ given by $\phi(t_i)=\gamma^{-1}(i)$ for every $i\in [n]$. The following is the main result of this section.

\begin{theorem}
\label{thm:polyhedral-description}
Let $G=(V, A)$ be an $n$-vertex permutation DAG and $\gamma$ be an umbrella-free topological ordering of its vertices. Let $Q(G, \gamma)$ be the convex-hull of indicator vectors of $\vec{x}\in \mathbb{N}^n$ whose sequence $\sigma:=(x_1, \ldots, x_n)$ is such that $\permutationGraph(\sigma) = (\{t_1, \ldots, t_n\}, A')$ is isomorphic to $G$ under the mapping $\phi:\{t_1, \ldots, t_n\} \rightarrow V$ given by $\phi(t_i) = \gamma^{-1}(i)$ for all $i\in [n]$. Then, 
\[
  Q(G, \gamma) = \left\{ x\in \R^{n}\ \middle\vert \begin{array}{l}
x_{\gamma(u)}\le x_{\gamma(v)} \ \ \ \ \ \ \ \forall (v,u)\in A, \\
x_{\gamma(v)}\le x_{\gamma(u)}-1 \  \ \forall (v,u)\not\in A \text{ with } \gamma(u)<\gamma(v), \text{ and}\\
x_i\ge 1 \ \ \ \ \ \ \ \ \ \ \ \ \ \ \  \forall\ i\in [n]
  \end{array}\right\}.
\]
\end{theorem}

For notational convenience, let $P(G, \gamma)$ denote the polyhedron defined in the RHS of Theorem \ref{thm:polyhedral-description}. 
Before proving Theorem \ref{thm:polyhedral-description}, we describe how $\alpha(G, \gamma)$ can be obtained by optimizing over $P(G', \gamma')$ for a graph $G'$ and an ordering $\gamma'$ obtained from $G$ and $\gamma$. 
Let $G'=(V', A')$ be obtained from $G$ by adding a vertex $t$ with edges $(t,u)$ for all $u\in V$ and $\gamma':V'\rightarrow [n+1]$ be defined as $\gamma'(u)=\gamma(u)$ if $u\in V$ and $\gamma'(t)=n+1$. We note that 
if $G$ is a permutation DAG and $\gamma$ is an umbrella-free topological ordering of $G$, then $G'$ is also a permutation DAG and $\gamma'$ is an umbrella-free topological ordering of $G'$. Moreover, we also have that 
\[
\alpha(G, \gamma) = \min\left\{x_{\gamma'(n+1)}: x\in Q(G', \gamma')\right\}.
\]
Thus, by Theorem \ref{thm:polyhedral-description}, the $\gamma$-alphabet size of $G$, i.e., $\alpha(G, \gamma)$, can be computed by optimizing along the objective direction $(0, \ldots, 0, 1)\in \R^{n+1}$ over the polyhedron $P(G', \gamma')$.

We now prove Theorem \ref{thm:polyhedral-description}. 
\begin{proof}[Proof of Theorem \ref{thm:polyhedral-description}]
We recall that a point $\vec{x}$ is an \emph{extreme point} of a polyhedron if $\vec{x}$ cannot be expressed as a convex combination of any two distinct points in the polyhedron. 
Any extreme point $x$ of $Q(G, \gamma)$ satisfies the constraints defining $P(G, \gamma)$. Thus, $Q(G, \gamma)\subseteq P(G, \gamma)$. In order to show equality, it suffices to show that all extreme points of $P(G, \gamma)$ are integral. 
Lemma \ref{lem:ext-point-integrality} shows that all extreme points of $P(G, \gamma)$ are integral, thus completing the proof of Theorem \ref{thm:polyhedral-description}. 
\end{proof}

\begin{lemma}
\label{lem:ext-point-integrality}
Let $G=(V, A)$ be an $n$-vertex DAG and $\gamma$ be a topological ordering of its vertices. If $\vec{x}$ is an extreme point of $P(G, \gamma)$, then $\vec{x} \in \mathbb{Z}^n$.
\end{lemma}
\begin{proof}
Suppose for contradiction that $\vec{x}$ is non-integral. We will show the existence of two points in $P(G, \gamma)$ such that $\vec{x}$ is a convex combination of these points. Let $S := \{i:x_i \not \in \mathbb{Z} \}$. We note that the set $S$ is non-empty due to our choice of $\vec{x}$. Let  $\epsilon \in \mathbb{R}$ be as follows
$$\epsilon := \min_{i \in S}\big\{\min(x_i - \floor{x_i}, \ceil{x_i} - x_i)\big\}.$$
Since $S$ is non-empty, we have $\epsilon > \frac{\epsilon}{2} > 0$. 
Let $\vec{y} \in \mathbb{R}^n$ be defined as follows:
$$y_i := \begin{cases}
\epsilon/2& \text{ if $i \in S$,}\\
0 & \text{ otherwise.}
\end{cases}$$
We note that $\vec{x} = \frac{1}{2}(\vec{x + y}) + \frac{1}{2}(\vec{x - y})$. 
It suffices to show that 
$\vec{x+y}, \vec{x-y} \in P$. 
We will show that the point $\vec{x+y} \in P$ and remark that the proof of $\vec{x-y} \in P$ is along very similar lines. We observe that $\vec{y}\ge 0$.

Constraint (3) is always satisfied as $x_i + y_i \geq x_i \geq 1$. 
We first focus on constraint $(1)$. Consider any arc $(v,u) \in A$. Since $\vec{y} \geq 0$, the constraint is easily seen to be satisfied in the cases where (1) $x_{\gamma(u)}, x_{\gamma(v)} \in \mathbb{Z}$; (2) $x_{\gamma(u)}, x_{\gamma(v)} \not \in \mathbb{Z}$; and (3) $x_{\gamma(u)} \in \mathbb{Z} \text{ but } x_{\gamma(v)} \not\in \mathbb{Z}$. Consider the case when $x_{\gamma(u)} \not\in \mathbb{Z}$ but $x_{\gamma(v)} \in \mathbb{Z}$. Then, we have that $$x_{\gamma(u)} + y_{\gamma(u)} < x_{\gamma(u)} + \epsilon \leq \ceil{x_{\gamma(u)}} \leq x_{\gamma(v)} = x_{\gamma(v)} + y_{\gamma(v)}.$$
The first inequality is by $y_i \leq \epsilon/2$ for all $i \in [n]$. The second inequality is by definition of $\epsilon$. The third inequality is due to $\vec{x} \in P$ and our case assumption that $x_{\gamma(v)} \in \mathbb{Z}$.  The equality relation is by definition of $\vec{y}$.

Next, we consider constraint (2). Let $\gamma(u) < \gamma(v)$ but $(v,u) \not \in A$. Similar to the above analysis, the constraint is easily seen to be satisfied in the cases where (1) $x_{\gamma(u)}, x_{\gamma(v)} \in \mathbb{Z}$; (2) $x_{\gamma(u)}, x_{\gamma(v)} \not \in \mathbb{Z}$; and (3) $x_{\gamma(u)} \not\in \mathbb{Z} \text{ but } x_{\gamma(v)} \in \mathbb{Z}$. Consider the case when $x_{\gamma(u)} \in \mathbb{Z}$ but $x_{\gamma(v)} \not\in \mathbb{Z}$.
Then, we have that $$x_{\gamma(v)} + y_{\gamma(v)} < x_{\gamma(v)} + \epsilon \leq \ceil{x_{\gamma(u)}} \leq x_{\gamma(u)} = x_{\gamma(u)} + y_{\gamma(u)}.$$
The first inequality is due to $y_i \leq \epsilon/2$ for all $i \in [n]$. The second inequality is by definition of $\epsilon$. The third inequality is due to $\vec{x} \in P$ and our case assumption that $x_{\gamma(u)} \in \mathbb{Z}$.  The equality relation is by definition of $\vec{y}$.
\end{proof}

Based on Lemma \ref{lem:ext-point-integrality}, it is natural to wonder if the integral extreme points of $P(G, \gamma)$ have any combinatorial interpretation when $G$ is an arbitrary DAG and $\gamma$ is an arbitrary topological ordering of $G$. 
The following lemma shows that integrality of $P(G, \gamma)$ is useful only when $G$ is a \emph{permutation} DAG and $\gamma$ is an \emph{umbrella-free} topological ordering of $G$.

\begin{lemma}\label{lem:polyhedron_feasibility}
Let $G$ be a DAG and $\gamma$ be a topological ordering of $G$. Then, $P(G,\gamma)$ is non-empty if and only if $G$ is a permutation DAG and $\gamma$ is umbrella-free.
\end{lemma}
\begin{proof}
The reverse direction follows from the correctness of \GreedyAssign (Lemma \ref{lem:GreedyAssign-Feasability}). We focus on proving the forward direction. Let $\vec{x} \in P(G, \gamma)$ be a feasible point. It suffices to show that $G$ is transitively closed and $\gamma$ is umbrella-free. 

First, assume for contradiction that $G$ is not transitively closed. Then, there exist arcs $(u,v), (v,w) \in A$ such that the arc $(u,w) \not \in A$. Since $\vec{x}$ is feasible, we have the following: (1) $x_{\gamma(v)} \leq x_{\gamma(u)}$; (2) $x_{\gamma(w)} \leq x_{\gamma(v)}$; and (3) $x_{\gamma(u)} \leq x_{\gamma(w)} - 1$. However, these inequalities do not admit any feasible solution, a contradiction. 

Next, assume for contradiction that $\gamma$ is not umbrella-free. Then, there exists a triple $(u,v,w)$ such that $\gamma(u) < \gamma(v) < \gamma(w)$, and the arc $(w,u) \in A$, but the arcs $(v, u), (w,v) \not \in A$. Since the point $\vec{x}$ is feasible, we have the following: $x_{\gamma(w)} \leq x_{\gamma(v)} - 1$; (2) $x_{\gamma(v)} \leq x_{\gamma(u)} -1$; and (3) $x_{\gamma(u)} \leq x_{\gamma(w)}$. However, these inequalities do not admit any feasible solution, a contradiction. 
\end{proof}
\section{An efficient algorithm for finding a maximum binary tree in bounded treewidth graphs}\label{sec:boundedtreewidth}

In this section, we prove Theorem~\ref{theorem:bounded-treewidth} by designing an efficient dynamic programming algorithm to find a maximum binary tree in graphs with bounded treewidth.
Given an undirected graph $G=(V,E)$, we say that a subgraph $T$, where $V(T) \subseteq V$ and $E(T) \subseteq E$, is a \emph{binary tree} in $G$ if $T$ is connected, acyclic, and $\deg_T(v) \le 3$ for every vertex $v \in V(T)$. The problem of interest is the following.

\begin{problem}{\umaxBT}
	Given: An undirected graph $G$.
	
	Goal: A binary tree in $G$ with maximum number of vertices. 
\end{problem}

We recall that if the maximum degree requirement for each vertex is two instead of three, then this problem is exactly the longest path problem. Our algorithm is modified from the standard dynamic programming approach for longest path in bounded treewidth graphs. Our main modification to the algorithm is to address the degree requirement for each vertex in the subproblems. We note that our technique is also applicable to find a maximum-sized bounded degree tree in a given graph of bounded treewidth, where the degree bounds are constant.
We will focus on the rooted version of the problem as defined below. 

\begin{problem}{\rumaxBT}
	Given: An undirected graph $G$ and a root $s \in V(G)$.
	
	Goal: A binary tree in $G$ containing $s$ where $\deg_T(s) \leq 2$ with maximum number of vertices.
\end{problem}

Now we introduce the definition of \emph{treewidth}.

\begin{definition} \label{def:tree-decomposition}
A {\em tree decomposition} of an undirected graph $G$ is given by $\cT = (T,\{X_t\}_{t \in V(T)})$, where $T$ is a tree in which every tree node $i \in V(T)$ is assigned a subset $X_i \subseteq V(G)$ of vertices of $G$, called a bag, such that the following conditions are satisfied:
\begin{enumerate}[label=(T\arabic*)]
\item \label{T1} $\cup_{i \in V(T)} X_i = V(G)$.
\item \label{T2} For each $\{u,v\} \in E(G)$, there is a tree node $i \in V(T)$ in the tree $T$ such that both $u$ and $v$ are in $X_i$.
\item \label{T3} For each vertex $u \in V(G)$, if there exist two tree nodes $i, j\in V(T)$ in the tree $T$ such that $u$ is in both $X_i$ and $X_j$, then $u\in X_k$ for every tree node $k \in V(T)$ on the unique path between $i$ and $j$.
\end{enumerate}
For ease of distinction, we will denote $i \in V(T)$ as a \emph{tree node} and the set $X_i$ as a \emph{bag} of the tree decomposition $\cal T$. The \emph{width} of a tree decomposition is $\max_{i \in V(T)}\{|X_i|-1\}$. The \emph{treewidth} $w_G$ of the graph $G$ is the minimum width among all possible tree decompositions. A tree decomposition with treewidth $w_G$ can be found in $(w_G 2^{w_G})^{O(w_G^2)}|V(G)|$ time \cite{bodlaender1996linear}.
\end{definition}

\subsection{Outline of the algorithm}
Our algorithm consists of two parts:
\begin{enumerate}
    \item Design a linear time algorithm to solve \rumaxBT.
    \item Reduce \umaxBT to \rumaxBT.
\end{enumerate}

For the first part, we will show the following theorem in the next sections.
\begin{restatable}{theorem}{rbddtw} \label{thm:rbddtw}
Let $s \in V(G)$ be a root. Given a tree decomposition of $G$ with width $w_G$, \rumaxBT can be solved in ${w_G}^{O(w_G)}|V(G)|$ time.
\end{restatable}

For the second part, by Theorem~\ref{thm:rbddtw}, given a tree decomposition of $G$ with width $w_G$, we can find a maximum binary tree of $G$ in $w_G^{O(w_G)}|V(G)|^2$ time by solving \rumaxBT for every possible choice of the root vertex in $G$. The time complexity can be brought down to linear as shown in the proof below. Corollary \ref{coro:bdd-MBT} proves Theorem \ref{theorem:bounded-treewidth}.

\begin{corollary}\label{coro:bdd-MBT}
Given a tree decomposition of $G$ with width $w_G$, \umaxBT can be solved in ${w_G}^{O(w_G)}|V(G)|$ time.
\end{corollary}

\begin{proof}
\umaxBT with $G$ as the input graph can be solved by the following reduction to \rumaxBT in a graph whose treewidth is one larger than the treewidth of $G$:
\begin{enumerate}
\item Construct $G'$ by adding a vertex $s$ and a binary tree $B$ such that:
\begin{enumerate}
\item $s$ is adjacent to all the vertices in $V(G)$,
\item $s$ is adjacent to the root of $B$, and
\item $|E(B)| = |E(G)|$.
\end{enumerate}
\item Solve \rumaxBT on $G'$ rooted at $s$ to obtain the maximum binary tree $T'$ of $G'$.
\item Let $v \in V(G)$ be the vertex such that $\{v,s\} \in E(T')$. Pick the subtree of $T'$ rooted at $v$.
\end{enumerate}
This is a generic reduction from \umaxBT to \rumaxBT (but not approximation-preserving). We note that the treewidth of the graph $G'$ is one larger than the treewidth of $G$.
\end{proof}

\subsection{Proof of Theorem \ref{thm:rbddtw}}
We will prove Theorem \ref{thm:rbddtw} in this section. We begin with a convenient form of tree decompositions and associated notations in Section \ref{subsubsec:special-tree-decompositions}. We describe the subproblems of the dynamic program along with the recursive expressions for these subproblems and analyze the run-time in Section \ref{subsubsec:bddtw-dp}. We prove the correctness of the recursive expressions in Section \ref{sec:dp_pf}.

\subsubsection{Special tree decompositions} \label{subsubsec:special-tree-decompositions}
In this subsection, we will introduce a modified form of tree decomposition that will be convenient for our dynamic program. We also introduce the associated notations and prove a crucial edge disjointness property (Lemma \ref{lem:disjoint}). 

Given a tree decomposition $\cT$ with an arbitrarily chosen root, the parent-child and ancestor-descendant relationships between the tree nodes are defined naturally. With these relationships, it will be convenient to think of \emph{nice tree decompositions} (defined below) as rooted trees.

\begin{definition} \label{def:nice-tree-decomposition}
A {\em nice tree decomposition} of an undirected graph $G$ is a tree decomposition $\cT = (T,\{X_t\}_{t \in V(T)})$ where $T$ is rooted at a tree node $r \in V(T)$ such that:

\begin{enumerate}
\item $X_l = \emptyset$ for every leaf tree node $l \in V(T)$ and $X_r = \emptyset$.
\item Every non-leaf tree node $i$ is one of the following four types:
\begin{description}
\item[Introduce vertex:] Tree node $i$ has only one child $j$ with $X_i = X_j \cup \{v\}$ for some $v \in V(G) \setminus X_j$. In this scenario, we say that vertex $v$ is \emph{introduced} at tree node $i$.
\item[Drop:] Tree node $i$ has only one child $j$ with $X_i = X_j \setminus \{v\}$ for some $v \in X_j$. In this scenario, we say that vertex $v$ is \emph{dropped} at tree node $i$.
\item[Join:] Tree node $i$ has two children $j$ and $k$ with $X_i=X_j=X_k$.
\item[Introduce edge:] Tree node $i$ is labeled by an edge $\set{u,v} \in E(G)$ such that both $u$ and $v$ belong to $X_i$, and $i$ has only one child $j$ with $X_i = X_j$. In this scenario, we say that edge $\set{u,v}$ is \emph{introduced} at tree node $i$.
\end{description}
\item Every edge of $G$ is introduced at exactly one tree node.
\item Every edge $\{u,v\}$ is introduced at tree node $i$ which is between tree nodes $j$ and $k$ such that:
\begin{enumerate}
\item Node $j$ is an ancestor of node $k$.
\item Node $j$ either drops $u$ or drops $v$.
\item All tree nodes between node $j$ and node $k$ are of type introduce edge and they introduce edges incident at the vertex that $j$ drops.
\item Node $k$ is not of type introduce edge.
\end{enumerate}
\end{enumerate}
\end{definition}

Given a tree decomposition with width $w_G$, a nice tree decomposition with $O(w_G |V(G)|)$ tree nodes can be computed in $O(w_G^2 \max\{|V(G)|,|V(T)|\})$ time \cite{FPT-book}. We note that by \ref{T3}, each vertex in $V(G)$ is dropped \emph{only once}, but may be introduced several times in a nice tree decomposition.

In order to solve \rumaxBT, we construct $G^s$ from $G$ by adding a pendant vertex $s'$ and edge $\{s,s'\}$. Formally, $G^s:=(V(G) \cup \{s'\}, E(G) \cup \set{\set{s,s'}})$.

\begin{observation}
The maximum binary tree in $G$ rooted at $s$ can be obtained by finding a maximum binary tree in $G^s$ rooted at $s'$ and removing the vertex $s'$ and the edge $\{s,s'\}$ from the tree.
\end{observation}

Now we introduce a special tree decomposition of $G^s$.

\begin{definition}
Given a nice tree decomposition of $G$, an \emph{$s'$-special tree decomposition} of $G^s$ is a tree decomposition obtained by the following steps:
\begin{enumerate}
\item Obtain a nice tree decomposition of $G$ with width $w_G$.
\item Add $s'$ to each bag and insert a tree node that introduces the edge $\set{s,s'}$ between the tree node that drops $s$ and its child.
\end{enumerate}
\end{definition}

The idea behind this tree decomposition is to ensure that every bag contains the new root $s'$, which is useful in the definition of the subproblems for the dynamic program.

For brevity, we will denote the resulting $s'$-special tree decomposition of $G^s$ as $\cT  = (T,\{X_t\}_{t \in V(T)})$. Let $r$ be the root tree node of $T$. Since we will only care about the optimal structure that is stored locally with respect to a tree node, we will use the following notations.

\begin{definition}
For a tree node $i$, let 
\begin{enumerate}
\item $S_i$ denote the set of tree nodes consisting of all its descendants (including $i$),
\item $V_i := \bigcup_{j \in S_i}{X_j}$ be the {\em descendant vertices}, 
\item $E_i := \set{e\in E(G): e \text{ is introduced at some tree nodes in the subtree rooted at $i$}}$ be the {\em descendent edges}, and 
\item $G_i := (V_i, E_i)$.
\end{enumerate}
\end{definition}

With the above notations, we show a \emph{disjointness} property for the join tree nodes.
\begin{lemma} \label{lem:disjoint}
In an $s'$-special tree decomposition, suppose a join tree node $i$ has two children $j$ and $k$ with $X_i=X_j=X_k$. Then $(V_j \setminus X_i) \cap (V_k \setminus X_i) = \emptyset$ and $E_j \cap E_k = \emptyset$.
\end{lemma}

\begin{proof}
First, we show that $(V_j \setminus X_i) \cap (V_k \setminus X_i) = \emptyset$. Assume for the sake of contradiction that there is a vertex $v \in (V_j \setminus X_i) \cap (V_k \setminus X_i)$. Then there is a descendent tree node $d_j$ of $j$, such that $v \in X_{d_j} \setminus X_i$. Similarly, there is a descendent tree node $d_k$ of $k$, such that $v \in X_{d_k} \setminus X_i$. From \ref{T3}, the vertex $v$ must belong to the bag of every tree node that is on the unique path between tree nodes $d_j$ and $d_k$, and this includes tree node $i$. Hence, $v \in X_i$, a contradiction.

Next, we show that $E_j \cap E_k = \emptyset$. By Definition~\ref{def:nice-tree-decomposition}, a tree node that drops vertex $v$ is followed by descendant tree nodes that introduce edges with $v$ as an endvertex until a descendant tree node that is not of type introduce edge is reached. Therefore, the edge $\{u,v\}$ where $u \in X_i$ and $v \in X_i$ is not yet included in $E_i$. Consequently, any edge in $E_j$ cannot have both endvertices in $X_j = X_i$, i.e. it must have an endvertex in $V_j \setminus X_j=V_j\setminus X_i$. Similar argument implies that any edge in $E_k$ cannot have both endvertices in $X_k=X_i$, i.e., it must have an endvertex in $V_k\setminus X_k=V_k\setminus X_i$. From $(V_j \setminus X_i) \cap (V_k \setminus X_i) = \emptyset$, we must have $E_j \cap E_k = \emptyset$.
\end{proof}

\subsubsection{The dynamic program} \label{subsubsec:bddtw-dp}

Our algorithm solves subproblems in a bottom-up fashion beginning with the leaf tree nodes. For every tree node $i$ and for all possible $(X,\mathcal{P},D)$, where 
\begin{enumerate}
\item $X \subseteq X_i$, 
\item $\mathcal{P}$ is a partition $\{P_1, P_2, ..., P_q\}$ of $X$ with each part of $\mathcal{P}$ being non-empty, and 
\item $D:X \rightarrow \{0,1,2,3\}$ is a function specifying the degree constraints on the vertices of $X$,  
\end{enumerate}
let $\mbt_i(X,\cP,D)$ be the \emph{maximum number of edges} in a binary forest $F$ satisfying the following five constraints:

\begin{enumerate} [label=(DP\arabic*),align=left]
\item \label{DP1} $F$ is a subgraph of $G_i$,
\item \label{DP2} $s' \in X$,
\item \label{DP3} $X_i \cap V(F) = X$, i.e. the set of vertices in $F$ from $X_i$ is exactly $X$,
\item \label{DP4} $F$ has \emph{exactly} $q$ connected components (trees) $T_1, ..., T_q$ such that $ V(T_p) \cap X = P_p$ for each $p \in \{1,2,...,q\}$, i.e. each part of $\mathcal{P}$ corresponds to a single tree of $F$, and 
\item \label{DP5} $\deg_F(v) = D(v)$ for each $v \in X$, i.e. each vertex $v\in X$ has exactly $D(v)$ edges incident to it in $F$.
\end{enumerate}

A forest is \emph{feasible} for $\mbt_i(X,\cP,D)$ if it satisfies properties \ref{DP1}, \ref{DP2}, \ref{DP3}, \ref{DP4}, and \ref{DP5}. If a forest has the maximum number of edges among the feasible forests, then we say that this forest is \emph{optimal} for $\mbt_i(X,\cP,D)$. If there is no feasible forest for $\mbt_i(X,\mathcal{P},D)$, then we call the subproblem to be \emph{infeasible}. For example, $\mbt_i(X,\mathcal{P},D)$ where $D(v)=0$ for a vertex $v \in X \subseteq X_i$ in a part $P\in \mathcal{P}$ with $|P| \geq 2$ is infeasible. We will define $\mbt_i(X,\mathcal{P},D)=-\infty$ when the subproblem is infeasible.

We show an example in Figure~\ref{fig:bbdtw-dp-ex}. Suppose at tree node $i$, $X_i = \{v_1, ..., v_6\}$, $V_i = \{v_1, ..., v_{10}\}$, and $G_i=(V_i,E_i)$ where $E_i$ consists of solid and dashed edges. All the edges in $E_i$ are introduced in the subtree rooted at $i$ where the solid edges are selected in the binary forest $F$ while the dashed edges are not.  Suppose the given input for tree node $i$ is $(X,P,D)$ where $X = \{v_1, ..., v_5\}$, $P = \set{\set{v_1,v_2},\set{v_3,v_5},\set{v_4}}$, $D(v_1)=D(v_2)=D(v_5)=1$, $D(v_4)=0$, and $D(v_3)=2$. In this case $OPT(i,X,P,D)=6$.

\begin{figure}[h]
\centering
\begin{tikzpicture}
\draw (0,0) ellipse (200pt and 25pt);
\node[vertex] (v1) at (-5,0) {$v_1$};
\node[vertex] (v2) at (-3,0) {$v_2$};
\node[vertex] (v3) at (-1,0) {$v_3$};
\node[vertex] (v4) at (1,0) {$v_4$};
\node[vertex] (v5) at (3,0) {$v_5$};
\node[vertex] (v6) at (5,0) {$v_6$};
\node[vertex] (v7) at (-4,-1.5) {$v_7$};
\node[vertex] (v8) at (-2,-1.5) {$v_8$};
\node[vertex] (v9) at (2,-1.5) {$v_9$};
\node[vertex] (v10) at (4,-1.5) {$v_{10}$};
\node at (-6,0) {$i$};
\node at (0,0) {$X$};
\draw[dashed] (-5.5,-0.5) -- (-5.5,0.5)
-- (3.5,0.5) -- (3.5,-0.5) --cycle;
%\draw[dashed] (0.5,-0.5) -- (0.5,0.5)
-- (5.5,0.5) -- (5.5,-0.5) --cycle;
\path
	(v1) edge (v7)
	(v2) edge (v7)
	(v3) edge (v8)
	(v3) edge (v9)
	(v5) edge (v9)
	(v9) edge (v10);
\path[dashed]
	(v2) edge (v8)
	(v4) edge (v9)
	(v7) edge (v8)
	(v6) edge (v10);
\end{tikzpicture}
\caption{$\mbt_i(X,\cP,D)$}
\label{fig:bbdtw-dp-ex}
\end{figure}

We recall that $r$ is the root of the $s'$-special tree decomposition of $G^s$. The number of edges in the maximum binary tree of $G^s$ rooted at $s'$ is exactly $\mbt_r(X=\{s'\},\mathcal{P}=\set{\set{s'}},D(s')=1)$. We now provide the recurrence relations for the dynamic programming algorithm for each type of tree node.\\

\noindent \textbf{Leaf:} Suppose $i$ is a leaf tree node. Then, 
\begin{equation*} \tag{DP-L} \label{eq:DP-L}
\mbt_i(X,\mathcal{P},D)=
\begin{cases}
0 &\text{if $X=\set{s'}, \cP=\set{\set{s'}}$, and $D(s')=0$,}\\
-\infty &\text{otherwise}.
\end{cases}
\end{equation*}
\mbox{}

\noindent \textbf{Introduce vertex:} Suppose $i$ introduces vertex $v$ and $j$ is the child of $i$ with $X_i=X_j\cup \{v\}$ for some $v\in V(G)\setminus X_j$. Consider the following conditions:
\begin{enumerate} [label=(IntroVertex \arabic*),align=left]
\item \label{IV1} $v\in X$ and either $\{v\}$ is not a part in $\cP$ or $D(v)>0$.
\item \label{IV2} $v \in X$, $\set{v}$ is a part in $\cP$, and $D(v)=0$.
\end{enumerate}
For $D:X\rightarrow \{0,1,2,3\}$, let $D':X\setminus \{v\}\rightarrow \{0,1,2,3\}$ be obtained by setting $D'(u)=D(u)$ for every $u\in X\setminus \{v\}$. Then, 
\begin{equation} \tag{DP-IV} \label{eq:DP-IV}
\mbt_i(X,\cP,D) =
\begin{cases}
-\infty & \text{if \ref{IV1},}\\
\mbt_j(X \setminus \set{v},\cP \setminus \set{\set{v}}, D') & \text{if \ref{IV2},}\\
\mbt_j(X,\cP,D) & \text{otherwise.}\\
\end{cases}
\end{equation}
\mbox{}

\noindent \textbf{Introduce edge:} Suppose $i$ introduces edge $\set{u,v}$ and $j$ is the child of $i$ with $X_i=X_j$. Consider the following condition:

\begin{enumerate} [label=(IntroEdge\arabic*),align=left]
\item \label{IE1} $u$ and $v$ are both in $X$ and in the same part $P\in \cP$.
\end{enumerate}

Suppose \ref{IE1} holds. Then let $\cP'$ be a partition of $X$ obtained from $\cP$ by splitting $P$ into two disjoint sets $P_u$ and $P_v$ such that $u \in P_u$ and $v \in P_v$, i.e. $\cP' = \cP \setminus \{P\} \cup \{P_u, P_v\}$, and $D':X\rightarrow \{0,1,2,3\}$ be obtained by setting $D'(w) = D(w)$ for $w \in X \setminus \set{u,v}$, $D'(u) = D(u)-1 \leq 3$, $D(u) \geq 1$, $D'(v) = D(v)-1 \leq 3$, and $D(v) \geq 1$. Then
\begin{equation} \tag{DP-IE} \label{eq:DP-IE}
\mbt_i(X,\cP,D) =
\begin{cases}
\max\{
\maxx_{\cP',D'}{\mbt_j(X,\cP',D')+1},
\mbt_j(X,\cP,D)\} & \text{if \ref{IE1},}\\
\mbt_j(X,\cP,D) & \text{otherwise.}
\end{cases}
\end{equation}

\noindent \textbf{Drop:} Suppose $i$ drops vertex $v$ and $j$ is the child of $i$ with $X_i = X_j \setminus \{v\}$ for some $v \in X_j$. Let $\cP'$ be a partition obtained by adding $v$ to one of the existing parts of $\cP$ and $D':X\cup\set{v} \rightarrow \set{0,1,2,3}$ is obtained by setting $D'(u)=D(u)$ for $u \in X$ and $D'(v)$ is set to some value in $\{1,2,3\}$. Then
\begin{equation} \tag{DP-D} \label{eq:DP-D}
\mbt_i(X,\cP,D) = 
\max\left\{
\maxx_{\cP',D'}\left\{\mbt_j(X \cup \set{v},\cP',D')\right\},
\mbt_j(X,\cP,D)
\right\}.
\end{equation}
\mbox{}

\noindent \textbf{Join:} Suppose $i$ has two children $j$ and $k$ with $X_i=X_j=X_k$. 

For arbitrary partitions $\cP^j, \cP^k$ of $X$, we obtain a graph $H(\cP^j,\cP^k)$ and a partition $\mathcal{Q}(\cP^j,\cP^k)$ as follows: for $b \in \{j,k\}$ we first construct an \emph{auxiliary graph} $H_b$ that contains vertices for every vertex in $X$ and vertices for every part in $\cP^b$, with the vertex corresponding to a part of $\cP^b$ being adjacent to all vertices that are in that part. We note that $V(H_j)\cap V(H_k)=X$ and $E(H_j)\cap E(H_k)=\emptyset$. We consider the \emph{merged graph} $H(\cP^j,\cP^k) := (V(H_j) \cup V(H_k), E(H_j) \cup E(H_k))$. Let $\mathcal{Q'}$ be a partition of $V(H(\cP^j,\cP^k))$ where vertices in the same connected component of $H$ belong to the same part of $\mathcal{Q'}$. Let $\mathcal{Q}(\cP^j, \cP^k)$ be the partition of $X$ obtained from $\mathcal{Q}'$ by restricting each part of $\mathcal{Q}'$ to elements in $X$. Namely, $\mathcal{Q}(\cP^j, \cP^k) = \{Q \cap X \mid Q \in \mathcal{Q'} \}$. We note that $\mathcal{Q}(\cP^j,\cP^k)$ can be found efficiently by running a depth first search on $H(\cP^j,\cP^k)$.

Let $(X,\cP,D)$ be the input at tree node $i$. Let $\cP^j$ and $\cP^k$ be partitions such that the graph $H(\cP^j,\cP^k)$ is a forest and the resulting partition $\mathcal{Q}(\cP^j,\cP^k)=\cP$. Let $D^j, D^k:X \rightarrow \set{0,1,2,3}$ be functions such that $D(v) = D^j(v)+D^k(v)$ for every $v \in X$. Then
\begin{equation} \tag{DP-J} \label{eq:DP-J}
\mbt_i(X,\cP,D) =
\begin{cases}
\maxx_{\cP^j,\cP^k,D^j,D^k}{\mbt_j(X,\cP^j,D^j)+\mbt_k(X,\cP^k,D^k),} \\
-\infty \text{ if there are no such $\cP^j,\cP^k,D^j$, and $D^k$.}
\end{cases}
\end{equation}

We show an example in Figure \ref{fig:join}. Suppose $X = \set{v_1,...,v_6}$, $\cP^j = \set{\set{v_1,v_2},\set{v_3,v_4},\set{v_5},\set{v_6}}$, and $\cP^k = \{\set{v_1},\{v_2,v_3,v_5\},\set{v_4},\set{v_6}\}$, the merged graph $H(\cP^j,\cP^k)$ is a forest. The resulting partition is $\{\set{v_1,v_2,v_3,v_4,v_5},\set{v_6}\}$.

\begin{figure}[h]
\centering
\begin{tikzpicture}
\draw (0,0) ellipse (220pt and 25pt);
\node[vertex] (v1) at (-5,0) {$v_1$};
\node[vertex] (v2) at (-3,0) {$v_2$};
\node[vertex] (v3) at (-1,0) {$v_3$};
\node[vertex] (v4) at (1,0) {$v_4$};
\node[vertex] (v5) at (3,0) {$v_5$};
\node[vertex] (v6) at (5,0) {$v_6$};
\node[vertex] (v7) at (-4,1.5) {};
\node[vertex] (v8) at (0,1.5) {};
\node[vertex] (v9) at (0,-1.5) {};
\draw[dashed] (-5.6,-0.5) -- (-5.6,2)
-- (5.5,2) -- (5.5,-0.5) --cycle;
\draw[dashed] (-5.5,-2) -- (-5.5,0.5)
-- (5.6,0.5) -- (5.6,-2) --cycle;
\node at (-2,1.5) {$H_j$};
\node at (2,-1.5) {$H_k$};
\node at (-6,0) {$i$};
\path
	(v1) edge (v7)
	(v2) edge (v7)
	(v3) edge (v8)
	(v4) edge (v8)
	(v2) edge (v9)
	(v3) edge (v9)
	(v5) edge (v9);
\end{tikzpicture}
\caption{$H(\cP^j,\cP^k)$}
\label{fig:join}
\end{figure}

When the root tree node $r$ is reached, the maximum number of edges of the binary tree in $G^s$ is $\mbt_r(X=\{s'\},P=\{\{s'\}\},D(s')=1)$. To find the optimal binary tree, we can simply backtrack through the subproblems to obtain an optimal solution.

This concludes the description of the algorithm. The proof of correctness is in Section~\ref{sec:dp_pf}. We proceed to analyze the running time. The bag of each tree node has at most $w_G+2$ vertices, so the number of subproblems per tree node is at most $(8w_G + 16)^{w_G + 2} = {w_G}^{O(w_G)}$, since for a tree node $i$ there are $2^{|X_i|}$ subsets $X \subseteq X_i$, at most $|X|^{|X|}$ partitions of $X$, and at most $4^{|X|}$ degree requirements. Solving a subproblem at a tree node requires considering at most all states from the children of that tree node, which takes $({w_G}^{O(w_G)})^2 = {w_G}^{O(w_G)}$ time. Thus, the running time for computing all the $\mbt$ values of a tree node is ${w_G}^{O(w_G)}$. We have thus proved Theorem~\ref{thm:rbddtw}.

\rbddtw*

\subsubsection{Correctness of the dynamic program} \label{sec:dp_pf}

We prove the correctness of the recursion expressions in the dynamic program described in the previous section.

\begin{lemma}
$\mbt_i(X,\cP,D)$ satisfies the recurrence relations given in \ref{eq:DP-L}, \ref{eq:DP-IV}, \ref{eq:DP-IE}, \ref{eq:DP-D}, and \ref{eq:DP-J}.
\end{lemma}
\begin{proof}
The proof consists of showing two parts: If $\mbt_i(X,\cP,D)$ is feasible, then the right hand side (RHS) of the recurrence formula should be the same as the left hand side (LHS), or equivalently, RHS is an upper bound and lower bound of LHS. In addition, if $\mbt_i(X,\cP,D)$ is infeasible, then the recurrence should return $\mbt_i(X,\cP,D)=-\infty$. We will show the upper and lower bound and the infeasible case with a proof by induction depending on the tree node type.\\

\noindent \textbf{Leaf:} For the base case, i.e. when $i$ is a leaf tree node, the only vertex in $G_i$ is $s'$. The only input for which $\mbt_i(X,\cP,D)$ is feasible is $X=\{s'\},\cP=\set{\set{s'}}$, and $D(s')=0$. For the feasible input, the node $s'$ must be included in the optimal binary forest $F$ otherwise constraint \ref{DP2} will be violated. Since no edges have been introduced in the subtree rooted at tree node $i$, the degree of $s'$ in the optimal binary forest $F$ will be zero. Therefore, $\mbt_i(X,\cP,D) = 0$. If $\mbt_i(X,\cP,D)$ is infeasible, then it is set to $-\infty$.\\

Suppose that $\mbt_j(X,\cP,D)$ is correct for all possible inputs $(X,\cP,D)$ for all children $j$ of a non-leaf tree node $i$. We will show that $\mbt_i(X,\cP,D)$ computed using the recurrence relation based on its children is correct. We consider the various types for tree node $i$.\\

\noindent \textbf{Introduce vertex:}
Suppose that $i$ has one child $j$ with $X_i = X_j \cup \{v\}$ for some $v \in V(G) \setminus X_j$. We recall that \emph{none} of the edges incident to $v$ have been introduced in the subtree rooted at $i$, so the vertex $v$ is isolated in $G_i$, i.e. $\deg_{G_i}(v)=0$. 
Suppose that $\mbt_i(X,\cP,D)$ is feasible with $F$ being an optimal binary forest. We have two cases:

\emph{Case 1:} Suppose that $v \in X$. Then $\set{v}$ must be a part of $\cP$ and $D(v)=0$ (i.e., condition \ref{IV2} is satisfied), otherwise $\mbt_i(X,\cP,D)$ is infeasible. If \ref{IV2} holds, then $(V(F) \setminus \{v\}, E(F))$ is  a feasible solution for $\mbt_j(X\setminus \{v\},\cP \setminus \{\{v\}\},D')$, where $D':X\setminus \{v\}\rightarrow \{0,1,2,3\}$ is obtained by setting $D'(u)=D(u)$ for every $u\in X\setminus \{v\}$. Hence, $\mbt_i(X,\cP,D) \leq \mbt_j(X \setminus \{v\},\cP \setminus \{\{v\}\},D')$. Since there is a feasible forest $F_j$ for $\mbt_j(X\setminus \{v\},\cP \setminus \{\{v\}\},D')$, a feasible forest for $\mbt_i(X,\cP,D)$ can be obtained by adding an isolated vertex $v$ to $V(F_j)$. Therefore, $\mbt_i(X,\cP,D) \geq \mbt_j(X \setminus \{v\},\cP \setminus \{\{v\}\},D')$.

\emph{Case 2:} Suppose $v \notin X$. Then $v$ is not in $F$. The forest $F$ is a feasible solution for $\mbt_j(X,\cP,D)$. Hence, $\mbt_i(X,\cP,D) \leq \mbt_j(X,\cP,D)$. Since there is a feasible forest $F_j$ for $\mbt_j(X,\cP ,D)$, and $F_j$ is also feasible for $\mbt_i(X,\cP,D)$, we have $\mbt_i(X,\cP,D) \geq \mbt_j(X,\cP,D)$.

Next, suppose that $\mbt_i(X,\cP,D)$ is infeasible. If \ref{IV1} holds, then we are done. Suppose \ref{IV1} fails, but \ref{IV2} holds. In this case, $\mbt_j(X\setminus \{v\},\cP \setminus \{\{v\}\},D')$, where $D':X\setminus \{v\}\rightarrow \{0,1,2,3\}$ is obtained by setting $D'(u)=D(u)$ for every $u\in X\setminus \{v\}$, is infeasible. Hence, $\mbt_i(X,\cP,D)=\mbt_j(X\setminus \{v\},\cP \setminus \{\{v\}\},D')=-\infty$. Suppose that both \ref{IV1} and \ref{IV2} fail. In this case, $\mbt_j(X,\cP,D)$ is infeasible. Hence, $\mbt_i(X,\cP,D)=\mbt_j(X,\cP,D)=-\infty$.\\

\noindent \textbf{Introduce edge:}
Suppose $i$ introduces edge $\set{u,v}$ and has one child $j$ such that $X_i = X_j$. Suppose that $\mbt_i(X,\cP,D)$ is feasible with $F$ being an optimal binary forest. We have two cases:

\emph{Case 1:} Suppose $u$ and $v$ are both in $X$ and also in the same part of $\cP$ (i.e., condition \ref{IE1} is satisfied). Then, the forest $F$ can either include or exclude $\{u,v\}$.

If $\{u,v\}$ is excluded, then $F$ is also feasible for $\mbt_j(X,\cP,D)$ and hence $\mbt_i(X,\cP,D) \leq \mbt_j(X,\cP,D)$. There is a feasible forest for $\mbt_j(X,\cP ,D)$ which is also feasible for $\mbt_i(X,\cP,D)$. Therefore, $\mbt_i(X,\cP,D) \geq \mbt_j(X,\cP,D)$.

If $\{u,v\}$ is included, then $\cP$ must have a part $P$ that can be partitioned into two disjoint parts $\cP_u$ and $\cP_v$ where $u \in \cP_u$ and $v \in \cP_v$. Hence, the forest $(V(F) \setminus \{u,v\}, E(F) \setminus \{\{u,v\}\})$ is also feasible for $\mbt_j(X,\cP',D')$ for some $\cP' = \cP \setminus \{P\} \cup \{P_u, P_v\}$ with $P=P_u\cup P_v$, and $D':X\rightarrow \{0,1,2,3\}$ obtained by setting $D'(w) = D(w)$ for $w \in X \setminus \set{u,v}$, $D'(u) = D(u)-1$, and $D'(v) = D(v)-1$. Consequently, $\mbt_i(X,\cP,D) \le \maxx_{\cP',D'}{\mbt_j(X,\cP',D')+1}$. On the other hand, a feasible forest for $\mbt_i(X,\cP,D)$ can be obtained by adding edge $\{u,v\}$ to a feasible forest for $\mbt_j(X,\cP',D')$, so $\mbt_i(X,\cP,D) \ge \maxx_{\cP',D'}{\mbt_j(X,\cP',D')+1}$.

\emph{Case 2:} Otherwise, either (1) $u \notin X$ or $v \notin X$, or (2) $u$ and $v$ are both in $X$ but not in the same part of $\cP$. In this case $\{u,v\}$ cannot be in $F$. Hence, the forest $F$ is a feasible solution for the subproblem $\mbt_j(X,\cP,D)$ and hence $\mbt_i(X,\cP,D) \leq \mbt_j(X,\cP,D)$. There is a feasible forest for $\mbt_j(X,\cP ,D)$ which is also feasible for $\mbt_i(X,\cP,D)$. Therefore, $\mbt_i(X,\cP,D) \geq \mbt_j(X,\cP,D)$.

Next, suppose that $\mbt_i(X,\cP,D)$ is infeasible. Then, $\mbt_j(X,\cP,D)$ is also infeasible and hence, $\mbt_j(X,\cP,D)=-\infty$. We consider \ref{IE1} now: In this case, for every partition $\cP' = \cP \setminus \{P\} \cup \{P_u, P_v\}$ with $P=P_u\cup P_v$, and for $D':X\rightarrow \{0,1,2,3\}$ obtained by setting $D'(w) = D(w)$ for $w \in X \setminus \set{u,v}$, $D'(u) = D(u)-1$, and $D'(v) = D(v)-1$, we have that $\mbt_j(X,\cP',D')$ is infeasible. Therefore $\maxx_{\cP',D'}{\mbt_j(X,\cP',D')}=-\infty$.\\

\noindent \textbf{Drop:}
Suppose $i$ drops $v$ and has a child $j$ such that $X_i = X_j \setminus \{v\}$ for some $v \in X_j$. Suppose that $\mbt_i(X,\cP,D)$ is feasible with $F$ being an optimal binary forest.

If $F$ excludes $v$, then $F$ is also a feasible solution for $\mbt_j(X,\cP,D)$ and hence $\mbt_i(X,\cP,D) \leq \mbt_j(X,\cP,D)$. There is a feasible forest for $\mbt_j(X,\cP ,D)$ which is also feasible for $\mbt_i(X,\cP,D)$. Therefore, $\mbt_i(X,\cP,D) \geq \mbt_j(X,\cP,D)$.

Suppose that $F$ includes $v$, then the forest $F$ is a feasible solution for some $\mbt_j(X \cup \set{v},\cP',D')$, where $\cP'$ is a partition obtained by adding $v$ to one of the existing parts of $\cP$ and $D':X\cup\set{v} \rightarrow \set{0,1,2,3}$ is obtained by setting $D'(u)=D(u)$ for $u \in X$ and $D'(v)$ is set to some value in $\{1,2,3\}$. Hence, $\mbt_i(X,\cP,D)\le \maxx_{\cP',D'}{\mbt_j(X \cup \set{v},\cP',D')}$. On the other hand, a feasible forest for $\mbt_i(X,\cP,D)$ can be obtained by selecting a feasible forest for $\mbt_j(X,\cP',D')$, so $\mbt_i(X,\cP,D) \ge \maxx_{\cP',D'}{\mbt_j(X \cup \set{v},\cP',D')}$.

Next, suppose that $\mbt_i(X,\cP,D)$ is infeasible. Then, $\mbt_j(X,\cP,D)$ is also infeasible and hence $\mbt_j(X,\cP,D)=-\infty$. Moreover, for every partition $\cP'$ obtained by adding $v$ to one of the existing parts of $\cP$ and every $D':X\cup\set{v} \rightarrow \set{0,1,2,3}$ obtained by setting $D'(u)=D(u)$ for $u \in X$ and every choice of $D'(v)\in \{1,2,3\}$, there is no feasible forest for $\mbt_j(X\cup \set{v},\cP',D')$. Hence, $\maxx_{\cP',D'}{\mbt_j(X \cup \set{v},\cP',D')}=-\infty$.\\

\noindent \textbf{Join:}
Suppose $i$ has two children $j$ and $k$ and $X_i=X_j=X_k$. Suppose that $\mbt_i(X,\cP,D)$ is feasible with $F$ being an optimal binary forest.

First we show that LHS $\leq$ RHS. By Lemma~\ref{lem:disjoint}, the edges in $F$ are from disjoint edge sets $E_j$ and $E_k$. Hence, the forest $F$ decomposes into two binary forests $F_j:=(V(F) \cap V_j ,E(F)\cap E_j)$ and $F_k:=(V(F) \cap V_k ,E(F)\cap E_k)$. In particular $F_j$ and $F_k$ are subgraphs of $G_j$ and $G_k$, respectively. Let $X^j := V(F_j) \cap X_j$ be the set of vertices in $X_j$ that belong to $F_j$ and similarly, let $X^k:=V(F_k)\cap X_k$.

\begin{claim} \label{cl:sameX}
$X^j = X^k = X$.
\end{claim}
\begin{proof}
We will show that $X^j=X$. The proof for $X^k=X$ will also follow by the same argument. 

We first have that $X \subseteq X^j$ because $X \subseteq V(F)$ and $X \subseteq X_i = X_j \subseteq V_j$ imply $X \subseteq V(F) \cap V_j \cap X_j = X^j$. Next, suppose that there exists a vertex $v \in X^j \setminus X$. Then $v \in X_j \setminus X = X_i \setminus X$. Besides, $v \in V(F_j) \subseteq V(F)$. However, these two statements violate \ref{DP3} so $X^j \subseteq X$.
\end{proof}

We show that there exist partitions $\cP^j$ and $\cP^k$ of $X$, degree requirements $D^j$, $D^k: X\rightarrow \{0,1,2,3\}$, and a decomposition of the forest $F$ into $F^j$ and $F^k$ such that
\begin{enumerate} [label=(J\arabic*)]
\item \label{J2} the merged graph $H(\cP^j,\cP^k)$ is a forest,
\item \label{J3} the partition $\mathcal{Q}(\cP^j,\cP^k)$ is exactly $\cP$,
\item \label{J4} $D^j(v) + D^k(v) = D(v)$ for $v \in X$, and
\item \label{J1} $F_j$ is feasible for $\mbt_j(X,\cP^j,D^j)$ and $F_k$ is feasible for $\mbt_k(X,\cP^k,D^k)$.
\end{enumerate}

Let $b \in \{j,k\}$. We obtain the partition $\cP^b$ of $X$ as follows: vertices of $X$ in the same connected component of $F_b$ will be in the same part of $\cP^b$, i.e., if $\mathcal{R}'_b$ is a partition of $V(F_b)$ based on the connected components of $F_b$, then $\cP^b=\{R'\cap X| R'\in \mathcal{R}_b'\}$.

We recall that $H_b$ is an auxiliary graph that contains vertices for every vertex in $X$ and vertices for every part in $\cP^b$, with the vertex corresponding to a part of $\cP^b$ being adjacent to all vertices that are in that part. The merged graph $H(\cP^j,\cP^k) := (V(H_j) \cup V(H_k), E(H_j) \cup E(H_k))$ where $V(H_j)\cap V(H_k)=X$ is such that $E(H_j)\cap E(H_k)=\emptyset$. Let $u, v \in X$. For brevity, we denote \emph{$u$-$v$ path} as a path between $u$ and $v$. From the fact that (1) $u$ and $v$ are in the same part of $\cP^j$ (or $\cP^k$) if and only if their is a unique $u$-$v$ path in $F_j$ (or $F_k$) and (2) $u$ and $v$ are in the same part of $\cP^j$ (or $\cP^k$) if and only if their is a unique $u$-$v$ path in $H_j$ (or $H_k$), we have the following observation:

\begin{observation}\label{obs:path}
Let $u, v \in X$, then there is a unique $u$-$v$ path in $F_j$ (or $F_k$) if and only there is a unique $u$-$v$ path in $H_j$ (or $H_k$).
\end{observation}

\begin{claim}
The merged graph $H(\cP^j,\cP^k)$ satisfies \ref{J2}. 
\end{claim}
\begin{proof}
The proof consists of showing two parts: (1) if $u, v \in X$ and there is a $u$-$v$ path in $H(\cP^j,\cP^k)$, then this path is unique, and (2) there is no cycle in $H(\cP^j,\cP^k)$. We will use the notation $u - F_j - v$ to denote that there is a unique $u$-$v$ path in $F_j$ and such a path may have intermediary vertices from $X \setminus \{u,v\}$. We will use similar notations for $F_k$, $H_j$, and $H_k$.

Let $u, v \in X$. Suppose that there is a $u$-$v$ path in $H(\cP^j,\cP^k)$. Since $E(H_j) \cap E(H_k) = \emptyset$ and $(V(H_j) \setminus X) \cap (V(H_k) \setminus X) = \emptyset$, the path must be one of the following:
\begin{enumerate}
\item $u - H_j - w_1 - H_k - w_2 - ... - w_{n-1} - H_k - w_n - H_j - v$,
\item $u - H_j - w_1 - H_k - w_2 - ... - w_{n-1} - H_j - w_n - H_k - v$,
\item $u - H_k - w_1 - H_j - w_2 - ... - w_{n-1} - H_k - w_n - H_j - v$, or
\item $u - H_k - w_1 - H_j - w_2 - ... - w_{n-1} - H_j - w_n - H_k - v$
\end{enumerate}
where $w_1, ..., w_n \in X \setminus \{u,v\}$. In short, a $u$-$v$ path in $H(\cP^j,\cP^k)$ must alternate edges between $H_j$ and $H_k$. By Observation~\ref{obs:path}, if the $u$-$v$ path in $H(\cP^j,\cP^k)$ is the first case, then there is a $u$-$v$ path $u - F_j - w_1 - F_k - w_2 - ... - w_{n-1} - F_k - w_n - F_j - v$ in $F$. The argument is similar for the other three cases. We note that the mapping of a $u$-$v$ path in $H(\cP^j,\cP^k)$ to a $u$-$v$ path in $F$ is a bijection. If there are two $u$-$v$ paths in $H(\cP^j,\cP^k)$, then there are two $u$-$v$ paths in $F$ and $F$ is not a forest, a contradiction. Therefore, we have the following observation:

\begin{observation} \label{obs:bij}
Let $u, v \in X$. There is no $u$-$v$ path in $F$ if and only if there is no $u$-$v$ path in $H(\cP^j,\cP^k)$. If there is a $u$-$v$ path in $F$, then there is a bijection between the unique $u$-$v$ path in $F$ and the unique $u$-$v$ path in $H(\cP^j,\cP^k)$.
\end{observation}

This implies that if $u, v \in X$ and there is a $u$-$v$ path in $H(\cP^j,\cP^k)$, then this path is unique.
Next, suppose that there is a cycle in $H(\cP^j,\cP^k)$. Then this cycle has even length in $H(\cP^j,\cP^k)$ since $H(\cP^j,\cP^k)$ is bipartite with $X$ on one side and vertices for every part in $\cP^j$ or $\cP^k$ on the other side. Furthermore, the cycle has length at least four and alternates vertices from $X$ and from parts in $\cP^j$ or $\cP^k$. This indicates that there must be two vertices $u,v \in X$ in the cycle, which contradicts the fact that if there is a $u$-$v$ path in $H(\cP^j,\cP^k)$, then such a path must be unique.
\end{proof}

\begin{claim}
The partition $\mathcal{Q}(\cP^j, \cP^k)$ satisfies \ref{J3}.
\end{claim}
\begin{proof}
Let $\mathcal{Q'}$ be a partition of $V(H(\cP^j,\cP^k))$ where vertices in the same connected component of $H(\cP^j,\cP^k)$ belong to the same part of $\mathcal{Q'}$. We recall that $\mathcal{Q}(\cP^j, \cP^k)$ is the partition of $X$ obtained from $\mathcal{Q}'$ by restricting each part of $\mathcal{Q}'$ to elements in $X$. That is, $\mathcal{Q}(\cP^j, \cP^k) = \{Q \cap X \mid Q \in \mathcal{Q'} \}$.

We prove \ref{J3} by showing that the following statements are equivalent for $u,v \in X$: 
\begin{enumerate} [label=(P\arabic*)]
\item \label{P1} $u$ and $v$ are in the same part of $\cP$,
\item \label{P2} there is a unique $u$-$v$ path in $F$,
\item \label{P3} there is a unique $u$-$v$ path in $H(\cP^j,\cP^k)$, and
\item \label{P4} $u$ and $v$ are in the same part of $\mathcal{Q}(\cP^j,\cP^k)$.
\end{enumerate}

\ref{P1}$\iff$\ref{P2}: This is by the definition of $\cP$. The vertices $u$ and $v$ are in the same part of $\cP$ if and only if they are in the same connected component of $F$. Since $F$ is a forest, the vertices $u$ and $v$ are in the same connected component of $F$ if and only if there is a unique $u$-$v$ path in $F$.

\ref{P2}$\iff$\ref{P3}: This is by Observation~\ref{obs:bij}.

\ref{P3}$\iff$\ref{P4}: Since $H(\cP^j,\cP^k)$ is a forest, there is a unique $u$-$v$ path in $H(\cP^j,\cP^k)$ if and only if $u$ and $v$ are in the same connected component of $H(\cP^j,\cP^k)$. By the definition of $\mathcal{Q}(\cP^j,\cP^k)$, $u$ and $v$ are in the same connected component of $H(\cP^j,\cP^k)$ if and only if they are in the same part of $\mathcal{Q}(\cP^j,\cP^k)$.

From \ref{P1}$\iff$\ref{P4}, we know that $\mathcal{Q}(\cP^j,\cP^k) = \cP$.
\end{proof}

For \ref{J4}, we set $D^j(v) := \deg_{F_j}(v)$ and $D^k(v) := \deg_{F_k}(v)$ for $v \in X$. By the edge disjointness of $F_j$ and $F_k$, we have $D^j(v) + D^k(v) = \deg_{F_j}(v) + \deg_{F_k}(v) = \deg_F(v) = D(v)$.

Now we show \ref{J1}.

\begin{claim}
$F_j$ is feasible for $\mbt_j(X,\cP^j,D^j)$ and $F_k$ is feasible for $\mbt_k(X,\cP^k,D^k)$.
\end{claim}
\begin{proof}
We will show that $F_j$ is feasible for $\mbt_j(X,\cP^j,D^j)$. The proof for $F_k$ being feasible for $\mbt_k(X,\cP^k,D^k)$ would follow by the same argument.

We need to show that properties \ref{DP1}, \ref{DP2}, \ref{DP3}, \ref{DP4}, and \ref{DP5} hold. \ref{DP1} holds because $F_j:=(V(F) \cap V_j ,E(F)\cap E_j)$ is a subgraph of $G_j$. \ref{DP2} holds since $s' \in X = X^j$ where $X^j := V(F_j) \cap X_j$. This is obtained from the fact that $F$ is feasible for $\mbt_i(X,\cP,D)$, which implies $s' \in X$, and $X^j = X$ by Claim~\ref{cl:sameX}. \ref{DP3} holds because of Claim~\ref{cl:sameX}. \ref{DP4} holds by the definition of $\cP^j$: vertices of $X$ in the same connected component of $F_j$ will be in the same part of $\cP^j$. \ref{DP5} holds by the definition $D^j(v) := \deg_{F_j}(v)$ for every $v \in X$.
\end{proof}

We have shown that there exist partitions $\cP^j$, $\cP^k$ of $X$, degree requirements $D^j, D^k: X\rightarrow \{0,1,2,3\}$, and a decomposition of the forest $F$ into $F^j$ and $F^k$ such that \ref{J2}, \ref{J3}, \ref{J4}, and \ref{J1} hold. Therefore, $\mbt_i(X,\cP,D) \leq \maxx_{\cP^j,\cP^k,D^j,D^k}{\mbt_j(X,\cP^j,D^j)+\mbt_k(X,\cP^k,D^k)}$.\\

Next, we show that LHS $\geq$ RHS. Given input $(X,\cP,D)$ at tree node $i$, suppose that we have partitions $\cP^j$ and $\cP^k$ of $X$, and degree requirements $D^j, D^k: X\rightarrow \{0,1,2,3\}$ such that
\begin{enumerate}
\item the merged graph $H(\cP^j,\cP^k)$ is a forest,
\item the partition $\mathcal{Q}(\cP^j,\cP^k)$ is exactly $\cP$,
\item $D^j(v) + D^k(v) = D(v)$ for $v \in X$, and
\item there are feasible forests $F_j$ for $\mbt_j(X,\cP^j,D^j)$ and $F_k$ for $\mbt_k(X,\cP^k,D^k)$. 
\end{enumerate}

We show that $F':=(V(F_j) \cup V(F_k), E(F_j) \cup E(F_k))$ is a feasible forest for $\mbt_i(X,\cP,D)$ with $\mbt_j(X,\cP^j,D^j) + \mbt_k(X,\cP^k,D^k)$ edges. This consists of showing the following: (1) $|E(F')| = \mbt_j(X,\cP^j,D^j) + \mbt_k(X,\cP^k,D^k)$, (2) $F'$ is a forest, and (3) $F'$ is feasible for $\mbt_i(X,\cP,D)$.

For the first part, we note that $F_j$ and $F_k$ are subgraphs of $G_j$ and $G_k$, respectively. By Lemma~\ref{lem:disjoint}, we have $E(G_j) \cap E(G_k) = \emptyset$ which implies $E(F_j) \cap E(F_k) = \emptyset$ and $|E(F')| = |E(F_j)| + |E(F_k)| = \mbt_j(X,\cP^j,D^j) + \mbt_k(X,\cP^k,D^k)$.

We now prove that $F'$ is a forest. Similar to the argument while proving LHS $\leq$ RHS, we note that Observation~\ref{obs:path} still holds.

\begin{claim}
$F'$ is a forest.
\end{claim}
\begin{proof}
The proof consists of showing two parts: (1) if $u, v \in X$ and there is a $u$-$v$ path in $F'$, then this path is unique, and (2) there is no cycle in $F'$.

Let $u,v \in X$. Since $E(F_j) \cap E(F_k) = \emptyset$ and $(V(F_j) \setminus X_i) \cap (V(F_k) \setminus X_i) = \emptyset$, if there is a $u$-$v$ path in $F'$, it must be one of the following:
\begin{enumerate}
\item $u - F_j - w_1 - F_k - w_2 - ... - w_{n-1} - F_k - w_n - F_j - v$,
\item $u - F_j - w_1 - F_k - w_2 - ... - w_{n-1} - F_j - w_n - F_k - v$,
\item $u - F_k - w_1 - F_j - w_2 - ... - w_{n-1} - F_k - w_n - F_j - v$, or
\item $u - F_k - w_1 - F_j - w_2 - ... - w_{n-1} - F_j - w_n - F_k - v$
\end{enumerate}
where $w_1, ..., w_n \in X \setminus \{u,v\}$. By Observation~\ref{obs:path}, if the $u$-$v$ path in $F'$ is the first case, then there is a $u$-$v$ path $u - H_j - w_1 - H_k - w_2 - ... - w_{n-1} - H_k - w_n - H_j - v$ in $H(\cP^j,\cP^k)$. The argument is similar for the other three cases. We note that the mapping of a $u$-$v$ path in $F'$ to a $u$-$v$ path in $H(\cP^j,\cP^k)$ is a bijection. If there are two $u$-$v$ paths in $F'$, then there are two $u$-$v$ paths in $H(\cP^j,\cP^k)$ and $H(\cP^j,\cP^k)$ is not a forest, a contradiction. Therefore, we have the following observation:

\begin{observation} \label{obs:bij2}
Let $u, v \in X$. There is no $u$-$v$ path in $H(\cP^j,\cP^k)$ if and only if there is no $u$-$v$ path in $F'$. If there is a $u$-$v$ path in $H(\cP^j,\cP^k)$, then there is a bijection between the unique $u$-$v$ path in $H(\cP^j,\cP^k)$ and the unique $u$-$v$ path in $F'$.
\end{observation}

This implies that if $u, v \in X$ and there is a $u$-$v$ path in $F'$, then this path is unique. Now, if $F'$ has a cycle, then this cycle must have two vertices $u,v \in X$ because $(V(F_j) \setminus X_i) \cap (V(F_k) \setminus X_i) = \emptyset$. This implies that there are two $u$-$v$ paths in $F'$, a contradiction.
\end{proof}

Finally, we show that $F'$ is feasible for $\mbt_i(X,\cP,D)$.

\begin{claim}
The forest $F'$ is feasible for $\mbt_i(X,\cP,D)$.
\end{claim}
\begin{proof}
We need to show \ref{DP1}, \ref{DP2}, \ref{DP3}, \ref{DP4}, and \ref{DP5}. \ref{DP1} holds because $F':=(V(F_j) \cup V(F_k), E(F_j) \cup E(F_k))$ is a subgraph of $G_i = (V(G_j) \cup V(G_k), E(G_j) \cup E(G_k))$. \ref{DP2} holds because we know $s' \in X$ from the feasibility of $\mbt_j(X,\cP^j,D^j)$ (or $\mbt_k(X,\cP^k,D^k)$). \ref{DP3} holds because $X_i = X_j = X_k$ and $X = X_i \cap V(F')$. To show $X = X_i \cap V(F')$, first $X = X_j \cap V(F_j) \subseteq X_i \cap V(F')$ and $X_i \cap V(F') = X_i \cap (V(F_j) \cup V(F_k)) = X_i \cap (X \cup (V(F_j) \setminus X) \cup (V(F_k) \setminus X)) = X_i \cap X \subseteq X$. \ref{DP5} holds by edge disjointness of $F_j$ and $F_k$ and the fact that $D(v) = D^j(v) + D^k(v)$ for every $v \in X$.

Now we prove \ref{DP4}, namely for $u,v \in X$, if $u$ and $v$ are in the same connected component of $F'$, then $u$ and $v$ are in the same part of $\cP$. We show that the following statements are equivalent for $u,v \in X$,
\begin{enumerate} [label=(P\arabic*$'$)]
\item \label{P1'} $u$ and $v$ are in the same connected component of $F'$,
\item \label{P2'} there is a unique $u$-$v$ path in $F'$,
\item \label{P3'} there is a unique $u$-$v$ path in $H(\cP^j,\cP^k)$, and
\item \label{P4'} $u$ and $v$ are in the same part of $\mathcal{Q}(\cP^j,\cP^k) = \cP$.
\end{enumerate}

\ref{P1'} $\iff$ \ref{P2'}: Since $F'$ is a forest, $u$ and $v$ are in the same connected component of $F'$ if and only if there is a unique $u$-$v$ path in $F'$.

\ref{P2'} $\iff$ \ref{P3'}: This is by Observation~\ref{obs:bij2}.

\ref{P3'} $\iff$ \ref{P4'}: This is by the definition of $\mathcal{Q}(\cP^j,\cP^k)$. Vertices $u$ and $v$ are in the same connected component of $H(\cP^j,\cP^k)$ if and only if they are in the same part of $\mathcal{Q}(\cP^j,\cP^k)$. Since $H(\cP^j,\cP^k)$ is a forest, there is a unique $u$-$v$ path in $H(\cP^j,\cP^k)$ if and only if $u$ and $v$ are in the same part of $\mathcal{Q}(\cP^j,\cP^k)$.
\end{proof}

Since $F'$ is a feasible forest with $\mbt_j(X,\cP^j,D^j) + \mbt_k(X,\cP^k,D^k)$ edges for $\mbt_i(X,\cP,D)$, we have $\mbt_i(X,\cP,D) \geq \maxx_{\cP^j,\cP^k,D^j,D^k}{\mbt_j(X,\cP^j,D^j)+\mbt_k(X,\cP^k,D^k)}$.

Next, suppose that $\mbt_i(X,\cP,D)$ is infeasible. For every $\cP^j$ and $\cP^k$ such that $H(\cP^j,\cP^k)$ is a forest, $\mathcal{Q}(\cP^j,\cP^k) = \cP$ is the resulting partition on $X$, and for every $D^j,D^k:X\rightarrow \{0,1,2,3\}$ such that $D(v) = D^j(v) + D^k(v)$ for every $v \in X$, either $\mbt_j(X,\cP^j,D^j)=-\infty$ or $\mbt_k(X,\cP^k,D^k)=-\infty$, otherwise we can find a feasible solution and this contradicts the assumption that $\mbt_i(X,\cP,D)$ is infeasible. If such $\cP^j,\cP^k,D^j$, and $D^k$ do not exist, then $\mbt_i(X,\cP,D)=-\infty$.
\end{proof}

\bibliography{references.bib}
\bibliographystyle{alpha}
\end{document}